\documentclass[a4paper,12pt,oneside]{article}
\usepackage{amsmath,amssymb,mathtools,bm,lipsum}
\usepackage[english]{babel}
\usepackage[T1]{fontenc}
\usepackage[utf8]{inputenc}
\usepackage{amstext}
\usepackage{amsthm}
\usepackage{bbm}
\usepackage{enumerate}
\usepackage{hyperref}
\usepackage[capitalise]{cleveref}
\usepackage{braket}

\usepackage[verbose=true,letterpaper]{geometry}
\AtBeginDocument{
	\newgeometry{
		textheight=9.5in,
		textwidth=7in,
		top=1in,
		headheight=14.5pt,
		headsep=10pt,
		footskip=20pt
	}
}

\def\titlename{\scshape Absence of Ground States in the Renormalized Massless Translation-Invariant Nelson Model}
\title{\LARGE\titlename}
\author{Thomas Norman Dam\\ Aarhus Universitet\footnote{Present Affiliation: Danske Bank, Superfly Analytics, QRA, Laksegade 11, 1063 Copenhagen K, Denmark}\\  \small Nordre Ringgade 1\\[-.5em]\small 8000 Aarhus C\\[-.5em] \small Denmark \and Benjamin Hinrichs\\ Friedrich Schiller Universit\"at Jena\\\small Ernst-Abbe-Platz 2\\[-.5em]\small  07743 Jena\\[-.5em] \small Germany\\ \small\texttt{benjamin.hinrichs@uni-jena.de}}
\newcommand{\shortauthors}{T.N. Dam, B. Hinrichs}
\date{June 24, 2022}

\usepackage{fancyhdr}
\fancyheadoffset{0pt}
\usepackage{titlesec}
\titleformat{\section}{\bfseries\scshape\Large}{\thesection}{1em}{}{}
\titleformat*{\subsection}{\scshape\bfseries\large}

\renewenvironment{thebibliography}[1]{
	\begin{oldthebibliography}{#1}
		\footnotesize
		\setlength{\itemsep}{0em}
		\setlength{\parskip}{0em}
	}
	{
	\end{oldthebibliography}
}

\newcommand{\IN}{\mathbb{N}}

\newcommand{\IR}{\mathbb{R}}
\newcommand{\IC}{\mathbb{C}}
\newcommand{\R}{\mathbb{R}}

\newcommand{\CC}{\mathbb{C}}

\newcommand{\NN}{\mathbb{N}}

\newcommand{\RR}{\mathbb{R}}

\newcommand{\HS}{\mathcal{H}}
\newcommand{\FS}{\mathcal{F}}

\newcommand{\cN}{\mathcal{N}}
\newcommand{\cB}{\mathcal{B}} 
\newcommand{\cC}{\mathcal{C}}
\newcommand{\cD}{\mathcal{D}} 
\newcommand{\cE}{\mathcal{E}}\newcommand{\cQ}{\mathcal{Q}}
\newcommand{\cF}{\mathcal{F}}
\newcommand{\cS}{\mathcal{S}}
\newcommand{\cH}{\mathcal{H}}

\newcommand{\cM}{\mathcal{M}}

\newcommand{\fH}{\mathfrak H}

\newcommand{\eps}{\varepsilon}
\newcommand{\ph}{\varphi}

\DeclareMathOperator*{\esssup}{ess\,sup}
\DeclareMathOperator*{\slim}{s-lim}
\DeclareMathOperator*{\wlim}{w-lim}

\numberwithin{equation}{section}
\newtheorem{thm}{Theorem}[section]

\newtheorem{lem}[thm]{Lemma}
\newtheorem{prop}[thm]{Proposition}

\theoremstyle{definition}

\newtheorem{rem}[thm]{Remark}
\newtheorem{ex}[thm]{Example}
\newtheorem{hyp}{Hypothesis}

\crefname{hyp}{Hypothesis}{Hypotheses}
\Crefname{hyp}{Hypothesis}{Hypotheses}
\crefname{lem}{Lemma}{Lemmas}
\Crefname{lem}{Lemma}{Lemmas}
\crefname{prop}{Proposition}{Propositions}
\Crefname{prop}{Proposition}{Propositions}
\crefname{enumi}{}{}
\Crefname{enumi}{}{}
\creflabelformat{enumi}{#2(#1)#3}
\crefname{equation}{}{}
\Crefname{equation}{}{}
\crefname{thm}{Theorem}{Theorems}
\Crefname{thm}{Theorem}{Theorems}

\newcommand{\partref}[2]{\hyperref[#2]{\cref*{#1} \cref*{#2}}}

\begin{document}

\maketitle\thispagestyle{empty}
\begin{abstract}
	\noindent
	We consider a model for a massive uncharged non-relativistic particle interacting with a massless bosonic field, widely referred to as the Nelson model. It is well known, that an ultraviolet renormalized Hamilton operator exists in this case. Further, due to translation-invariance, it decomposes into fiber operators.
	
	In this paper, we treat the renormalized fiber operators. We give a description of the operator and form domains and prove that the fiber operators do not have a ground state. Our results hold for any non-zero coupling constant and arbitrary total momentum. Our proof for the absence of ground states is a new generalization of methods recently applied to related models. A major enhancement we provide, is that the method can be applied to models with degenerate ground state eigenspaces.
	
	\bigskip
	\noindent
	{\footnotesize {\em Keywords:} Nelson Model, Energy Renormalization, Translation Invariance, Ground State Absence
	
	}
\end{abstract}

\section{Introduction}
In this paper, we prove absence of ground states for the fiber operators in the translation-invariant renormalized massless Nelson model. The method used in the proof follows along the lines of \cite{Dam.2018}, but crucially avoids non-degeneracy of ground states as an ingredient, even though it would be available (see \cite{Miyao.2019}). Instead, we rely on a momentum estimate for hypothetical ground states derived via rotation-invariance. The purpose of avoiding non-degeneracy of ground states as an ingredient is threefold: First of all, the method found in this paper can in principle be applied to translation-invariant models which cannot be expected to have non-degenerate ground states, such as models with spin. Second, it allows us to remove unnatural technical assumptions and work under very general hypotheses. Third, it results in an almost self-contained proof of our claim.  

In order to prove the main result, we need some quite strong technical results about the convergence to the renormalized models. These also allow us to prove several new results about the domain and the form domain of the renormalized fiber operators, which hold both for the massive and the massless model. In particular, we show that all fiber operators have the same form domain, but the intersection between the domain of two different fiber operators is $\{0\}$ in case of the three-dimensional massless Nelson model.

\bigskip\noindent
To further introduce the results, we will use vague but standard definitions of the operators. For rigorous definitions we refer the reader to \cref{sec:model}.

The Nelson model was originally treated by Edward Nelson in his paper \cite{Nelson.1964}. He investigated a formal expression for a Hamiltonian describing the interaction of a chargeless particle with a scalar bosonic field. We here denote the dispersion relation of the bosonic field as $\omega:\IR^3\to\IR$ and the interaction function between particle and field as $v:\IR^3\to \IR$. In the physical case, they are given as
\begin{equation}\label{eq:physical}
\omega(k) = \sqrt{\mu^2+k^2} \qquad \mbox{and}\qquad v= g\omega^{-1/2}
\end{equation}
for some coupling constant $g\in\R\setminus\{0\}$ and the photon mass $\mu\ge 0$.
By removing the interaction between the particle and field modes with momentum larger than a parameter $\Lambda\in (0,\infty)$, the so-called ultraviolet cutoff, Nelson was able to define the corresponding operator acting on $\psi\in L^2(\IR^3,\FS)$ as $$H_\Lambda\psi(x) = -\Delta_x\psi(x) + d\Gamma(\omega)\psi(x) + \int\limits_{|k|<\Lambda}e^{-ikx}v(k)(a_k^\dag+a_{-k})\psi(x) dk.$$
Here, $-\Delta_x$ denotes the positive Laplacian on $L^2(\IR^3,\FS)$, $\FS$ denotes the bosonic Fock space over $L^2(\IR^3)$, $d\Gamma$ is the second-quantization of a multiplication operator,  and $a_k$, $a^\dag_k$ are the pointwise annihilation and creation operators.
He then found a constant $E_\Lambda$, sometimes referred to as the self-energy, such that $H_\Lambda+E_\Lambda$ converges to an operator $H_\infty$ in the strong resolvent sense as $\Lambda\to\infty$. $H_\infty$ is the operator we refer to as the translation-invariant renormalized Nelson Hamiltonian. It was later proven in \cite{Cannon.1971}, that the convergence is actually in the norm resolvent sense.

The operators $H_\Lambda$ and $H_\infty$ strongly commute with the total momentum
$$ P = -i\nabla + d\Gamma(m),\quad\mbox{where}\ m:\IR^3\to\IR^3,\ k\mapsto k, $$
which is the reason they are called translation-invariant. This implies that there is a unitary operator $V$ (sometimes called the Lee-Low-Pines operator, see \cite{LeeLowPines.1953}) such that the operator $V H_\Lambda V^*$ decomposes into a direct integral
$$ V H_\Lambda V^* = \int_{\IR^3}^\oplus H_\Lambda(\xi)d\xi \quad\mbox{for all}\ \Lambda\in(0,\infty], $$
where the so-called fiber operators $H_\Lambda(\xi)$ act on $\FS$. For $\Lambda<\infty$ they can explicitly be written as
$$ H_\Lambda(\xi) = (\xi-d\Gamma(m))^2 + d\Gamma(\omega) + \int\limits_{|k|<\Lambda}v(k)(a^\dag_k+a_k)dk. $$
In \cite{Cannon.1971} it was also proven, that $H_\Lambda(\xi)+E_\Lambda$ converges to $H_\infty(\xi)$ in the norm resolvent sense as $\Lambda\to\infty$ for all $\xi\in \RR^3$.

Results about the domain of the full operator $H_\infty$ are derived in \cite{GriesemerWuensch.2018}, whereas we analyze the domains of the renormalized fiber operators $H_\infty(\xi)$ (cf. \cref{thm:operatordomain}). Explicitly, we prove that the form domain of the renormalized fiber operators $H_\infty(\xi)$ is independent of the total momentum $\xi$. We further derive a necessary and sufficient condition for the operator domain to also be independent of $\xi$. If this is not the case, we show that the domains of $H_\infty(\xi_1)$ and $H_\infty(\xi_2)$ only have the zero-vector in common if $\xi_1\ne\xi_2$. This especially holds for the physically relevant case \cref{eq:physical}, independent of coupling constant or photon mass. \cref{LargeFiberProp} in \cref{sec:fiberrenorm} also contains numerous interesting results about improved convergence of $H_\Lambda(\xi)$ as $\Lambda\to\infty$, which hold in very high generality. Actually, the domain properties can be seen as almost direct consequences of those results.
In the recent papers \cite{LampartSchmidt.2019,Posilicano.2020,Schmidt.2021}, explicit expressions for the renormalized operators have been derived using interior boundary conditions. However, the regularity assumptions therein are considerably stronger than the ones we use.
Therefore, our approach to the renormalization of the fiber operators uses the results on the full operators from \cite{GriesemerWuensch.2018} and extracts the behaviour of the fiber operators using properties of direct integrals, discussed in \cref{app:directintegrals}.  In principle, the same results could be obtained solely working on the fiber operators, by adapting the arguments from \cite{GriesemerWuensch.2018} to the fiber operators directly. We do not use this approach here, because it would mean a lengthy repetition of methods known from \cite{Nelson.1964,GriesemerWuensch.2018}, but refer to \cite{Hinrichs.2022} for a proof along these lines.

We say the Nelson model is infrared divergent, if $v/\omega$ is not square-integrable in any ball centered at zero. In the physically relevant case \cref{eq:physical}, this holds if we set the photon mass $\mu=0$. In \cref{thm:nogs}, we prove that in case of infrared divergence the renormalized operator $H_\infty(\xi)$ does not have a ground state for any $\xi\in\IR^3$, i.e., $\inf \sigma(H_\infty(\xi))$ is not an eigenvalue. The proof is based on the techniques developed in the papers \cite{Dam.2018, HaslerHerbst.2008a}, but extends it to a more general setting. In \cite{HaslerHerbst.2008a}, absence of ground states for fiber operators of the translation-invariant non-relativistic Pauli-Fierz model is proven under the condition that the mass shell $\Sigma_\Lambda(\xi)=\inf(\sigma(H_\Lambda(\xi)))$ is differentiable and the derivative is non-zero. This approach was later adapted to the semi-relativistic Pauli-Fierz model in the paper \cite{KoenenbergMatte.2014}, where the authors also rely on a non-zero derivative of the mass shell. In fact, one could also easily adapt the proof in \cite{HaslerHerbst.2008a} to obtain absence of ground states in the Nelson model, whenever the mass shell is differentiable. The differentiability condition is very hard to prove, and so far, has only been proven for weak coupling and a restricted set of total momenta $\xi\in\IR^\nu$ (see \cite{FrohlichPizzo.2010,KoenenbergMatte.2014,HaslerHerbst.2011a}). In \cite{Dam.2018}, the differentiability is replaced with non-degeneracy of the eigenspace associated with $\Sigma_\Lambda(\xi)$, rotation invariance of the mass shell and an HVZ type theorem. This leads to a simple, non-perturbative proof for absence of ground states, which holds for any coupling and all $\xi\in \RR^\nu$. This method, however, could never apply to the models treated in \cite{HaslerHerbst.2008a,KoenenbergMatte.2014}, as those fiber operators have natural degeneracies in their ground states due to the non-zero spin. In this paper, we give a comparatively short proof for absence of ground states for the renormalized fiber operators without relying on non-degeneracy of ground states. It should be noted that non-degeneracy results are available \cite{Miyao.2018,Miyao.2019,Lampart.2020} for the Nelson model. This, in principle, would allow us to extend the method from \cite{Dam.2018} to the renormalized case more directly. However, this would come with the following disadvantages. First, it would not help generalize the method from \cite{Dam.2018} to handle a broader class of models. Second, it would add extra unnatural technical conditions. Third, the present article would be a lot less self-contained. The main technical advancement in comparison to the proof in \cite{Dam.2018} can be found in the second half of \cref{sec:proof}, starting with \cref{lem:U-invertible}. It reflects in the fact that we treat the vector $V(\xi)$ defined in \cref{def:UVECT} as a vector of operators and not as a vector of scalars.

Absence of ground states in the full Nelson model has already been proven in the infrared divergent case, see \cite{DerezinskiGerard.2004,HiroshimaMatte.2019}. On the other hand, when $v/\omega$ is square integrable, there are several results about existence of ground states, see \cite{Gerard.2000,GubinelliHiroshimaLorinczi.2014} for the full Nelson model and \cite{Frohlich.1973,Frohlich.1974,Pizzo.2003} for the fiber operators. Our result therefore fits well into the general picture that local square-integrability of $v/\omega$ is critical for the existence of ground states in the Nelson model. In the context of scattering theory, absence of ground states for the fiber operators poses some challenges, as ground states near $\xi=0$ are used to define scattering states, which is why results of this type are often referred to as the infrared catastrophe. One workaround is to construct ground states in a so called non-equivalent Fock representation, which is done in \cite{BachmannDeckertPizzo.2012,Pizzo.2003}. 

\section{Notation, Model and  Main Result}\label{sec:model}

In this \lcnamecref{sec:model} we first introduce Fock space notation and operators in \cref{sec:notation}. Then, we define the Nelson model in \cref{sec:nelson} and state our main results \cref{thm:operatordomain} and \cref{thm:nogs}.\vspace*{-1em}

\subsection{Fock Space Operators}
\label{sec:notation}

Throughout this paper we assume $\nu\in\IN$ and write $\cH=L^2(\IR^\nu)$ for the state space of a single boson.
Let $\cF$ be the bosonic Fock space defined by
\begin{equation}\label{defn:Fockspace}
\cF=\IC\oplus\bigoplus_{n=1}^\infty \FS^{(n)} \qquad\mbox{with}\ \FS^{(n)}= L^2_{sym}(\IR^{n\nu}),
\end{equation} 
where we symmetrize over the $n$ $\IR^\nu$-variables in each component. We write an element $\psi\in \cF$ in terms of its coordinates $\psi=(\psi^{(n)})$ and define the vacuum $\Omega=(1,0,0,\dots)$. 

\noindent For measurable functions $\omega:\IR^\nu\to\IR$ we define
\begin{equation}\label{defn:dGamma-scalar}
d\Gamma(\omega) =0\oplus\bigoplus_{n=1}^{\infty} \omega^{(n)}\qquad\mbox{with}\  \omega^{(n)}(k_1,\ldots,k_n) = \sum_{i=1}^{n}\omega(k_i).
\end{equation}
Further, for unitary operators $U$ on $\HS$ we define
\begin{align}
\Gamma(U)&=1\oplus\bigoplus_{n=1}^\infty  U^{\otimes n}
\end{align}
as operators on $\FS$. We will write $\Gamma^{(n)}(U)$ for the restriction to $\FS^{(n)}$. 

For $g\in \cH$, we define the exponential vector $\epsilon(g)\in\FS(\HS)$ by $\epsilon(g)^{(0)}=1$ and
\begin{equation}\label{defn:corherentstate}
\epsilon(g)^{(n)}(k_1,\ldots,k_n)=\frac{1}{\sqrt{n!}}\prod_{i=1}^ng(k_i).
\end{equation}
Also for $g\in \cH$ we define the annihilation operator $a(g)$ and creation operator $a^{\dagger}(g)$ using $a(g)\Omega=0$, $a^\dagger(g)\Omega=g$ and for $f\in\FS^{(n)}$
\begin{align}
a(g)f (k_1,\ldots,k_{n-1})&={\sqrt{n}}\int \overline{g(k)}f(k,k_1,\ldots,k_{n-1}) dk,\label{defn:annihilation}\\
a^\dagger(g)f(k_1,\ldots,k_n,k_{n+1})&=\frac{1}{\sqrt{n+1}}\sum_{i=1}^{n+1}g(k_i)f(k_1,\ldots,\widehat{k}_i,\ldots,k_{n+1}), \label{defn:creation}
\end{align}
where $\widehat{k}_i$ means that $k_i$ is omitted from the argument. One can show that these operators can be extended to closed operators on $\cF$ that satisfy $(a(g))^*=a^{\dagger}(g)$. From the creation and annihilation operator, we define the symmetric operator
\begin{equation}\label{defn:fieldoperator}
\varphi(g)=\overline{ a(g)+a^\dagger(g) }.
\end{equation}
For $h\in\HS$ there exists a unique unitary map $W(h)$, called Weyl operator, such that
\begin{equation}\label{defn:Weyloperator}
W(h)\epsilon(g)=e^{-\frac 12\|h\|^2-\braket{h,g}}\epsilon(h+g) \qquad\mbox{for all}\ g\in\cH.
\end{equation}
Some well-known properties of the above operators are collected in the next \lcnamecref{Prop:Bregninger}.
\begin{lem}\label{Prop:Bregninger}
	Let $\omega:\IR^\nu\to\IR$ be a measurable function, $U$ be an unitary operator on $\HS$ and $f,g\in \cH$. Then
	\begin{enumerate}[(1)]
		\item $d\Gamma(\omega)$ is selfadjoint. If $\omega\ge 0$, then $d\Gamma(\omega)\geq 0$.
		\item $\varphi(g)$ is selfadjoint and $e^{i\varphi(g)}=W(-ig)$.
		\item $\Gamma(U)$ is unitary with $\Gamma(U)^* = \Gamma(U^*)$ and
		\begin{align*}
		\Gamma(U)\varphi(g)\Gamma(U)^*&=\varphi(Ug),\\
		\Gamma(U)W(f)\Gamma(U)^*&=W(Uf),\\
		\Gamma(U)d\Gamma(\omega)\Gamma(U)^*&=d\Gamma(U\omega U^*).
		\end{align*}
		\item\label{part:secondquantizedupperbounds} Assume $\omega>0$ almost everywhere and $g\in \cD(\omega^{-\frac{1}{2}})$. Then $\varphi(g)$ and $a(g)$ are $d\Gamma(\omega)^{1/2}$-bounded and we have the bounds
		\begin{align*}
		\lVert a(g) \psi \lVert&\leq  \lVert \omega^{-\frac{1}{2}}g \lVert  \lVert d\Gamma(\omega)^{\frac{1}{2}}\psi \lVert \quad\mbox{and}\\
		\lVert \varphi(g) \psi \lVert&\leq 2 \lVert (\omega^{-\frac{1}{2}}+1)g \lVert  \lVert (d\Gamma(\omega)+1)^{\frac{1}{2}}\psi \lVert, 
		\end{align*}
		which hold for $\psi\in\cD(d\Gamma(\omega)^{\frac{1}{2}})$. Especially, $\varphi(g)$ is infinitesimally $d\Gamma(\omega)$-bounded. Furthermore, $d\Gamma(\omega)+\varphi(g)\geq -\lVert \omega^{-\frac{1}{2}}g \lVert^2$.
		\item\label{part:Weylstrongcontinuous} $f\mapsto W(f)$ is strongly continuous and $W(f)W(g)=e^{-i\operatorname{Im}(\langle f,g \rangle)}W(f+g)$.
	\end{enumerate}
\end{lem}
\begin{proof}
	These results can for example be found in \cite{LorincziHiroshimaBetz.2011,Parthasarathy.1992,Arai.2018}.
\end{proof}
\noindent
If we have a measurable function $h:\RR^\nu\to\RR^p$ with coordinate functions $h_1,...,h_p$, it follows that $d\Gamma(h_i)$ and $d\Gamma(h_j)$ have commuting unitary groups for any $j\neq i$. For a function $f:\RR^p\to \RR$ we will then write
\begin{equation}\label{eq:defining_eqn}
f(d\Gamma(h)):=f(d\Gamma(h_1),...,d\Gamma(h_p))=0\oplus\bigoplus_{n=1}^{\infty} f(h^{(n)})
\end{equation}
where $h^{(n)}(k_1,...,k_n)=h(k_1)+...+h(k_n)$. We will need the following \lcnamecref{Lem:propertiesOfFuncDGamma}.
\begin{lem}\label{Lem:propertiesOfFuncDGamma}
	Let $h:\RR^\nu\rightarrow \RR^p$ be measurable with coordinate functions $h_1,...,h_p$. For all $\xi \in \RR^p$ and $s>0$ we have $\cD(\lvert d\Gamma(h)-\xi \lvert^s)=\bigcap_{i=1}^\nu\cD(\lvert d\Gamma(h_i) \lvert^s)$. Further, if $\omega:\RR^\nu\rightarrow \RR$ and $\lvert h(k)\lvert \leq \omega(k)$ for all $k\in \RR^\nu$ then for all $s\in[0,1]$ the operator $\lvert d\Gamma(h)\lvert^s$ is $d\Gamma(\omega)$-bounded.
\end{lem}
\begin{proof}
	For all $x,\xi_1,\xi_2\in\IR^p$ we have the inequalities
	$$ |x-\xi_1|^{2s} \le 2^{2s}(|x-\xi_2|^{2s} + |\xi_2-\xi_1|^{2s}) \qquad\mbox{and}\qquad |x_i|^{2s}\le|x|^{2s}\le p^{2s}\sum_{i=1}^{p}|x_i|^{2s}. $$
	Together with the spectral theorem this shows $\cD(\lvert d\Gamma(h)-\xi \lvert^s)$ is independent of $\xi$ and $\cD(\lvert d\Gamma(h)\lvert^s)=\bigcap_{i=1}^\nu\cD(\lvert d\Gamma(h_i) \lvert^s)$. Fix $s\in [0,1]$ and note that
	$$
	\lvert h^{(n)}(k_1,..,k_n) \lvert^s\leq \lvert h^{(n)}(k_1,..,k_n) \lvert+1\leq \omega^{(n)}(k_1,..,k_n)+1.
	$$
	This proves $\lvert d\Gamma(h)\lvert^s$ is $d\Gamma(\omega)$-bounded.
\end{proof}

\subsection{The Translation-Invariant Nelson Model}\label{sec:nelson}
To define the Nelson model, we fix two measurable functions $\omega:\RR^\nu\rightarrow  \RR$ and $v:\RR^\nu\rightarrow  \RR$. For each $\Lambda\ge 0$ we define $v_\Lambda(k)=1_{  \{\lvert k\lvert\leq \Lambda  \} }v(k)$.
\begin{hyp}\label{hyp1}
	We say $\omega$ and $v$ satisfy \cref{hyp1} if the following holds.
	\begin{enumerate}[(1)]
		\item\label{part:strictlypositive} $\omega>0$ almost everywhere.
		\item\label{part:UVbound} $v_\Lambda\in \cD(\omega^{-1/2})$ for all $\Lambda\geq 0$.
	\end{enumerate}
\end{hyp}
\noindent
Let $V_x$ denote the unitary map on $L^2(\RR^\nu)$ defined as multiplication by $k\mapsto e^{ik\cdot x}$. The Nelson model with ultraviolet cutoff $\Lambda$ is defined as
\begin{align}\label{defn:fulloperator}
H_\Lambda=-\Delta\otimes 1+1\otimes d\Gamma(\omega)+\int_{\IR^\nu}^{\oplus} \ph(V_{-x}v_\Lambda) dx
\end{align}
on $L^2(\IR^\nu,\FS)$. For a detailed definition and properties of direct integrals, we refer the reader to \cref{app:directintegrals}. Note,  that the dominated convergence theorem implies strong continuity of the map $x\mapsto V_x$, so the direct integral makes sense by \cref{Prop:Directint}.

Let $m:\IR^\nu\to\IR^\nu,\ m(k):=k.$
Given $\xi\in \RR^\nu $ we then define the fiber operators
\begin{align}\label{defn:fiber}
H_\Lambda(\xi)&=\lvert \xi-d\Gamma(m)\lvert^2+d\Gamma(\omega)+\varphi(v_\Lambda)
\end{align}
on $\FS$. The following \lcnamecref{prop:operatorsareSA} is well-known.
\begin{prop}\label{prop:operatorsareSA}
	Assume \cref{hyp1} holds. Then for all $\Lambda\ge 0$ and $\xi\in\IR^\nu$
	\begin{enumerate}[(1)]
		\item \label{part:FullisSA} $H_\Lambda$ is selfadjoint on $\cD(H_0)$ and bounded below,
		\item \label{part:fiberdomain} $H_\Lambda(\xi)$ is selfadjoint on $\cD(H_0(0))=\cD(d\Gamma(\omega))\cap\cD(\lvert d\Gamma(m)\lvert^2)$ and bounded below.
	\end{enumerate}
\end{prop}
\begin{proof}
	The operators $H_0$ and $H_0(0)$ are sums of selfadjoint nonnegative commuting operators and hence nonnegative and selfadjoint. By the Kato-Rellich theorem (cf. \cite[Theorem X.12]{ReedSimon.1975}) and \cref{Prop:Bregninger,Prop:Directint} we find $H_\Lambda$ and $H_\Lambda(\xi)$ are selfadjoint on the domains $\cD(H_0)$ and $\cD(H_0(\xi))$, respectively. That $\cD(H_0(\xi))=\cD(H_0(0))$ follows from \cref{Lem:propertiesOfFuncDGamma}. 
\end{proof}
\noindent
The connection between the fiber operators and the full model is described in the following \lcnamecref{Prop:LLP}, which goes back to \cite{LeeLowPines.1953}.
\begin{prop}\label{Prop:LLP}
	Define the operator (also known as the Lee-Low-Pines operator)
	\begin{align}\label{defn:LLP}
	V= (F\otimes 1)^* \int_{\IR^\nu}^{\oplus} \Gamma(V_x)dx,
	\end{align}
	where $F$ denotes the Fourier transform on $L^2(\IR^\nu)$.
	If \cref{hyp1} holds and $\Lambda\in [0,\infty)$, then
	\begin{equation}\label{eq:fiberdecomposition}
	VH_\Lambda V^*= \int_{\IR^\nu}^{\oplus} H_\Lambda(\xi)d\xi.
	\end{equation}
\end{prop}
\noindent
We now want to renormalize the model, i.e., remove the ultraviolet cutoff $\Lambda$. To that end, we need the following assumptions.

\begin{hyp}\label{hyp2}
	We say $\omega$ and $v$ fulfill \cref{hyp2}, if the following holds.
	\begin{enumerate}[(1)]
		\item\label{part:IRbound} There is $\sigma>0$ such that $m_\sigma:=\inf_{\lvert k\lvert \geq \sigma}\omega(k)>0$,
		\begin{equation*}
		\int_{\{\lvert k\lvert> \sigma  \}}\frac{\lvert v(k)\lvert^2}{\omega(k)^{1/2}(1 + \lvert k\lvert^2)}dk<\infty \,\,\,\,\,\, \text{and} \,\,\,\,\,\, \int_{\{\lvert k\lvert> \sigma  \}}\frac{\lvert v(k)\lvert^2\omega(k)}{(1 + \lvert k\lvert^2)^2} dk<\infty.
		\end{equation*} 
		\item $v(k)=v(-k)$ and $\omega(k)=\omega(-k)$ for all $k\in \RR^\nu$. 
		\item\label{part:omegabound} $k\mapsto \omega(k)(1+\lvert k\lvert^2)^{-1}$ is bounded. 
	\end{enumerate}
\end{hyp}
\noindent
Now define
\begin{align}\label{defn:selfenergy}
E_{\Lambda}=\int_{\lvert k\lvert \leq \Lambda}\frac{\lvert v(k)\lvert^2}{\omega(k)+\lvert k\lvert^2}dk,
\end{align}
which is finite by assumption \cref{hyp1} \cref{part:UVbound}. The self-energy renormalization of the Nelson model goes back to \cite{Nelson.1964,Cannon.1971} and is stated in the following \lcnamecref{prop:renormalization}.

\begin{prop}\label{prop:renormalization}
	Assume \cref{hyp1,hyp2} hold. Then
	\begin{enumerate}[(1)]
		\item\label{part:fulloperatorrenormalizable} There is a selfadjoint and lower-bounded operator $H_\infty$ such that $H_\Lambda+E_\Lambda$ converges to $H_\infty$ in the norm resolvent sense as $\Lambda\to\infty$.
		\item For all $\xi\in\IR^\nu$ there is a selfadjoint and lower-bounded operator $H_\infty(\xi)$ such that $H_\Lambda(\xi)+E_\Lambda$ converges to $H_\infty(\xi)$ in the norm resolvent sense as $\Lambda\to\infty$.\\ Further, $\xi\mapsto \inf\sigma(H_\infty(\xi))$ is continuous and uniformly bounded below.  
		\item The map $\xi\mapsto H_\infty(\xi)$ is continuous in the norm resolvent sense and the decomposition \cref{eq:fiberdecomposition} holds for $\Lambda=\infty$.
	\end{enumerate}
\end{prop}
\noindent
The domain of the full renormalized operator $H_\infty$ is thoroughly studied in \cite{GriesemerWuensch.2018}. We review their results in \cref{sec:fiberrenorm} and derive similar statements for the renormalized fiber operators $H_\infty(\xi)$ in \cref{LargeFiberProp}. For the case $\Lambda<\infty$ a straightforward calculation from the definition \cref{defn:fiber} leads to the transformation law
\begin{equation}
\label{eq:fiber-transformation}
H_\Lambda(\xi_2) = H_\Lambda(\xi_1) -  2 d\Gamma((\xi_2-\xi_1)\cdot m) + 2(\xi_2-\xi_1)\cdot \xi_1 + |\xi_2-\xi_1|^2
\end{equation}
for $\xi_1,\xi_2 \in\IR^\nu$. However, in the $\Lambda=\infty$ case, this statement heavily depends on the domains of the involved operators. 

\noindent For a selfadjoint operator $A$ we define the associated sesquilinear form as
\begin{align}\label{defn:form}
q_A(\psi,\phi)=\braket{|A|^{1/2}\psi,\operatorname{Sign}(A)|A|^{1/2}\phi}
\qquad\mbox{for all}\
\psi,\phi \in\cQ(A)=\cD(\lvert A\lvert^{1/2}).
\end{align}
We prove the following result on the renormalized fiber operators.

\begin{thm}\label{thm:operatordomain}
	Assume \cref{hyp1,hyp2} hold and let $\sigma$ as in \partref{hyp2}{part:IRbound}. 
	For $K\ge \sigma$, we define $B_K\in\HS$ as
	$$ B_K(k) = 1_{\{K\le|k|\}}(k)\frac{v(k)}{\omega(k)+|k|^2}. $$
	The following domain and transformation statements hold:
	\begin{enumerate}[(1)]
		\item\label{part:quadraticdomain} The form domain of $H_\infty(\xi)$ is independent of $\xi\in\IR^\nu$. Explicitly, $$\cQ({H_\infty(\xi)}) = W(B_K)^*\cD(H_0(0)^{1/2})\qquad\mbox{for any}\ K\ge\sigma.$$ Further, $\cQ({H_\infty(\xi)})\subset \cD(d\Gamma(\omega)^{1/2})\cap \cD(\lvert d\Gamma(m) \lvert^{2/3})$ and the transformation rule \cref{eq:fiber-transformation} with $\Lambda=\infty$ holds in the sense of sesquilinear forms for all $\xi_1,\xi_2\in\IR^\nu$.
		\item\label{part:operatordomain} For $\xi_1,\xi_2\in\IR^\nu$ the operator domains satisfy $\cD(H_\infty(\xi_1))=\cD(H_\infty(\xi_2))$ if and only if $k\mapsto (\xi_2-\xi_1)\cdot k B_K(k)$ is square-integrable. In the affirmative case, the transformation rule \cref{eq:fiber-transformation} with $\Lambda=\infty$ holds. Otherwise, the following equality holds $\cD(H_\infty(\xi_1))\cap\cD(H_\infty(\xi_2))=\{0\}$.
	\end{enumerate}
\end{thm}

\noindent
For our main result we add a third hypothesis. 
\begin{hyp}\label{hyp3}
	We say $\omega$ and $v$ fulfill \cref{hyp3}, if the following holds.
	\begin{enumerate}[(1)]
		\item\label{part:increasing} $\omega$ is continuous and $|k_1|>|k_2|$ implies $\omega(k_1)>\omega(k_2)$.
		\item $\omega$ and $v$ are rotation invariant.
		\item $v_\Lambda\notin \cD(\omega^{-1})$ for one and hence all $\Lambda>0$.
		\item\label{part:lim} The limit $C_\omega = \lim\limits_{k\to 0}\dfrac{|k|}{\omega(k)}\in(0,\infty)$ exists.
	\end{enumerate}
\end{hyp}
\noindent
For $\Lambda\in [0,\infty]$ let
$
\Sigma_\Lambda(\xi)=\inf(\sigma(H_\Lambda(\xi))).
$
We obtain the following theorem.
\begin{thm}
	\label{thm:nogs}
	Assume \cref{hyp1,,hyp2,,hyp3} hold and $\nu\geq 2$. Then $\Sigma_\infty( \xi)$ is not an eigenvalue of $H_\infty(\xi)$ for any $\xi\in \RR^\nu$.
\end{thm}
\begin{ex}
	We refer to the physical case of a (possibly massive) uncharged non-relativistic particle interacting with a radiation field in $\nu=3$ dimensions. In this case $\omega(k)=\sqrt{\mu^2+k^2}$ and $v=g\omega^{-1/2}$. We note that \cref{hyp1,hyp2} are satisfied for any photon mass $\mu\ge 0$ and coupling constant $g\ne 0$. One easily checks that $(\xi_2-\xi_1)\cdot kB_K(k)$ is not square-integrable for any choice of $\xi_1\ne\xi_2$ and $K>0$. Hence, by \cref{thm:operatordomain}, $\cD(H_\infty(\xi_1))\cap\cD(H_\infty(\xi_2))=\{0\}$ if $\xi_1\ne\xi_2$.
	
	The physical model is infrared divergent in the massless case $\mu= 0$, i.e., it then satisfies \cref{hyp3}. By \cref{thm:nogs} there is no ground state in this case. Hence, a massive non-relativistic particle interacting with a quantized massless bosonic field does not exhibit a stable ground state, due to infrared divergence.
\end{ex}

\section{Fiber Renormalization}\label{sec:fiberrenorm}
In this section we go into more details on the renormalization procedure for the fiber operators. The central result we obtain is \cref{LargeFiberProp}, which gives a collection of regularity results on the renormalized fiber operators. We then use these to prove \cref{thm:operatordomain}. We begin by recalling the renormalization procedure outlined in \cite{GriesemerWuensch.2018}, as we will need some of their results in our proofs.

Throughout this section, we assume that \cref{hyp1,,hyp2} hold and define $\sigma$ as in \partref{hyp2}{part:IRbound}. 
For $ \sigma \leq K<\Lambda\leq \infty$, we define the map
\begin{align}\label{defn:BKL}
B_{K,\Lambda}(k)=1_{\{K\leq \lvert k \lvert\leq \Lambda\}}(k)\frac{v(k)}{\omega(k)+\lvert k\lvert^2}.
\end{align}
We collect some regularity properties of these functions for future reference.
\begin{lem}\label{lem:propertiesOfB}Let $\sigma\le K<\Lambda\le\infty$.
	Then the following holds:
	\begin{enumerate}[(1)]
		\item\label{part:BKLfinite} $B_{K,\Lambda}\in \cD(\omega^a)\cap\cD(|m|^b)\cap \cD(\omega^a|m|^b)$ for all $a,b\in \RR$
		\item\label{part:BKinfty} $B_{K,\infty}\in\cD(\omega^{1/2}) \cap \cD(\omega^{-1/2}|m|)\cap\cD(|m|^{s})$ for all $s\leq \frac{2}{3}$ 
	\end{enumerate}
\end{lem}
\begin{proof}
	\cref{part:BKLfinite} holds, since $B_{K,\Lambda}$ is compactly supported and $\omega\in L^\infty_{loc}(\IR^\nu)$ by \partref{hyp2}{part:omegabound}. 
	Now \cref{part:BKinfty} follows due to the inequality
	\begin{equation*}
	|m(k)|^{2/3} \leq \omega(k)^{1/2}1_{\{\lvert m\lvert^{2/3} \leq \omega^{1/2}\}}(k)+| m(k) | \omega(k)^{-1/4}1_{\{\lvert m\lvert^{2/3}>  \omega^{1/2}\}}(k),
	\end{equation*}
	the integrability conditions in \partref{hyp2}{part:IRbound} and \partref{hyp2}{part:omegabound}. 
\end{proof}

\noindent
Since the Weyl representation is strongly continuous, the direct integral
\begin{align}\label{defn:UKL}
U_{K,\Lambda}:= \int_{\IR^\nu}^{\oplus} W( V_{-x}B_{K,\Lambda}) dx
\end{align}
exists. Using this, we define
\begin{align}\label{defn:HKL}
H_{K,\Lambda}=U_{K,\Lambda}H_{\Lambda}U_{K,\Lambda}^*+E_{\Lambda}.
\end{align}
Now, the following \lcnamecref{ConvTransformed} is \cite[Pages 5-8]{GriesemerWuensch.2018}.
\begin{prop}\label{ConvTransformed}
	For all $\sigma\le K<\Lambda<\infty$, we have $$\cD(H_{K,\Lambda})=U_{K,\Lambda}\cD(H_0)=\cD(H_0).$$ Moreover, there is a symmetric sesquilinear form $Q_{K,\Lambda}$ on $\cQ({H_0})$ such that
	\begin{align*}
	q_{H_{K,\Lambda}}=q_{H_0}+Q_{K,\Lambda}.
	\end{align*}
	\begin{enumerate}[(1)]
		\item\label{part:uliged1} For any $\varepsilon>0$ there are $K,b>0$ such that
		\begin{equation*}\label{eq:uliged1}
		\lvert Q_{K,\Lambda}(\psi,\psi)\lvert   \leq \varepsilon q_{H_0}(\psi,\psi)+b\|\psi\|^2 \qquad\mbox{for all}\ \Lambda>K\ \mbox{and}\ \psi\in \cD(H_0^{1/2}).
		\end{equation*}
		\item\label{part:convlambda0bound} For any $K\ge\sigma$ and $\varepsilon>0$ there is $\Lambda_0>K$ such that
		\begin{equation*}\label{eq:uliged2}
		\lvert Q_{K,\Lambda}(\psi,\psi)-Q_{K,\Lambda'}(\psi,\psi)\lvert \leq \varepsilon q_{H_0}(\psi,\psi) \qquad \mbox{for}\ \Lambda,\Lambda'>\Lambda_0\ \mbox{and}\ \psi\in \cD(H_0^{1/2}).
		\end{equation*}
		\item\label{part:conv} There exists a sesquilinear form $Q_{K,\infty}$ on $\cD(H_0^{1/2})$ such that
		\begin{equation*}
		Q_{K,\infty }(\psi,\phi)=\lim_{\Lambda\rightarrow \infty} Q_{K,\Lambda}(\psi,\phi) \qquad\mbox{for}\ \phi,\psi\in\cD(H_0^{1/2}).
		\end{equation*}
		The bounds in \cref{eq:uliged1,eq:uliged2} are also satisfied for $\Lambda=\infty$. 
	\end{enumerate}
\end{prop}
\noindent
Griesemer and Wünsch use the above \lcnamecref{ConvTransformed} to deduce the existence of the norm-resolvent limit $H_{K,\infty}$ in \cite[Theorem 3.3]{GriesemerWuensch.2018}. We add some further technical properties to their results in the next \lcnamecref{ConvFullOp}. From now on, we will write the resolvents of $H_{K,\Lambda}$ with $K\ge\sigma$ and $\Lambda\in(K,\infty]$ as
\begin{equation}\label{defn:resolvent-HKL}
R_{K,\Lambda}(z) = (H_{K,\Lambda}-z)^{-1} \qquad\mbox{for any}\ z\notin\sigma(H_{K,\Lambda}).
\end{equation}
\begin{prop}\label{ConvFullOp}
	Fix any $K\ge\sigma$.
	\begin{enumerate}[(1)]
		\item\label{part:Grosstransformedfullconvergence} The operators $H_{K,\Lambda}$ are bounded below uniformly in $\Lambda>K$.
		As $\Lambda\to\infty$ they converge to a lower-bounded selfadjoint operator $H_{K,\infty}$ in the norm resolvent sense.
		Further, $\cQ({H_{K,\infty}})=\cQ({H_0})$ and $q_{H_{K,\infty}}=q_{H_0}+Q_{K,\infty}$.
		\item\label{part:GrossFullSqrtConv} 
		For $\lambda<\inf\sigma(H_{K,\infty})$, there is $\Lambda_0>K$ such that $\lambda<\inf\sigma(H_{K,\Lambda})$ for $\Lambda\ge\Lambda_0$.\\ Setting \[C_{K,\Lambda}(\lambda):=(H_0+1)^{1/2}R_{K,\Lambda}(\lambda)^{1/2}\in\cB(L^2(\IR^\nu,\FS)),\] we have
		\begin{align*}
		\lim_{\Lambda\rightarrow \infty} C_{K,\Lambda}(\lambda)C_{K,\Lambda}(\lambda)^*=C_{K,\infty}(\lambda)C_{K,\infty}(\lambda)^*.
		\end{align*}
		\item\label{part:Fulloperatorconv} The operators $H_{\Lambda}+E_\Lambda$ are uniformly bounded below and converge to the operator $H_{\infty}=U_{K,\infty} H_{K,\infty} U_{K,\infty}^*$ in the norm resolvent sense as $\Lambda\to\infty$.	
	\end{enumerate}	
\end{prop}
\begin{proof}
	We first prove the uniform lower bounds in \cref{part:Grosstransformedfullconvergence,part:Fulloperatorconv}. Pick $K',b>0$ as in \partref{ConvTransformed}{part:uliged1} 
	corresponding to $\varepsilon=1/2$. Then for $\Lambda>K'$ we have
	\begin{equation*}
	q_{H_{K',\Lambda}}=q_{H_{0}}+Q_{K,\Lambda} \geq -b
	\end{equation*}
	showing $H_{K',\Lambda}$ is uniformly bounded below by $-b$. As $\inf\sigma(H_{K,\Lambda})=\inf\sigma(H_{\Lambda})+E_\Lambda$ for all $K<\Lambda$, it remains only to prove $H_\Lambda$ is uniformly bounded below on $[0,K']$. However, this directly follows from \cref{Prop:Directint}.
	
	The convergence in \cref{part:Grosstransformedfullconvergence,part:Fulloperatorconv} is covered in \cite{GriesemerWuensch.2018} for all $K$ larger than some $K'>\sigma$. Hence, we may assume they hold for $K'$ and only need to prove the statement for all $ K\in [\sigma,K')$. Observe that for $\Lambda\in [K',\infty]$ we have $U_{K,\Lambda}=U_{K,K'}U_{K',\Lambda}=U_{K',\Lambda}U_{K,K'}$ by \cref{Prop:Bregninger}, so
	\begin{equation}\label{eq:trasvig}
	H_{K,\Lambda}=U_{K,K'}U_{K',\Lambda} (H_{\Lambda}+E_\Lambda)U_{K',\Lambda}^* U_{K,K'}^*=U_{K,K'}H_{K',\Lambda}U_{K,K'}^*
	\end{equation}
	holds for all $\Lambda\in[K',\infty)$. Hence $H_{K,\Lambda}$ converges to $H_{K,\infty}=U_{K,K'}H_{K',\infty}U_{K,K'}^*$ in the norm resolvent sense. Using this, we also obtain
	\begin{align*}
	H_\infty &= U_{K',\infty}^* H_{K',\infty} U_{K',\infty} \\&=U_{K',\infty}^*U_{K,K'}^*U_{K,K'} H_{K',\infty}U_{K,K'}^*U_{K,K'} U_{K',\infty}\\&= U_{K,\infty}^* H_{K,\infty} U_{K,\infty}.
	\end{align*}
	As $U_{K,K'}\cD(H_0)=\cD(H_0)$ by \cref{ConvTransformed}, we see from \cref{lem:domainsqrt} that $U_{K,K'}\cQ({H_{0}})=\cQ({H_{0}})$. Now, since $\cQ({H_{K',\infty}})=\cQ({H_{0}})$ by assumption on $K'$, we obtain $\cQ({H_{K,\infty}})=U_{K,K'}\cQ({H_{K',\infty}})=\cQ({H_{0}})$.
	For $\phi,\psi \in \cQ(H_{K,\infty})$ we use \cref{eq:trasvig} and the fact $H_{K,\infty}=U_{K,K'}H_{K',\infty}U_{K,K'}^*$ to obtain
	\begin{align*}
	q_{H_{K,\infty}}(\phi,\psi)&=q_{H_{K',\infty}}(U_{K,K'}^*\phi,U_{K,K'}^*\psi)\\&= \lim_{\Lambda\rightarrow  \infty } q_{H_{K',\Lambda}}(U_{K,K'}^*\phi,U_{K,K'}^*\psi)\\&=\lim_{\Lambda\rightarrow  \infty } q_{H_{K,\Lambda}}(\phi,\psi)\\&=q_{H_0}(\phi,\psi)+Q_{K,\infty}(\phi,\psi),
	\end{align*}
	which concludes the proofs of \cref{part:Grosstransformedfullconvergence,part:Fulloperatorconv}.
	
	It remains only to prove \cref{part:GrossFullSqrtConv}, so let $\lambda<\inf\sigma(H_{K,\infty})$. By \cref{lem:normresolventConv} $\lim\limits_{\Lambda\to\infty}\inf\sigma(H_{K,\Lambda})=\inf\sigma(H_{K,\infty})$, so there is $\Lambda_0>K$ such that $\lambda<\inf\sigma(H_{K,\Lambda})$ for all $\Lambda\ge\Lambda_0$.
	We set $Z=\lVert H_0^{1/2} R_{K,\infty}(\lambda)^{1/2}\lVert$. For $\eps>0$, we choose $\delta_\eps$ such that
	\begin{equation*}
	Z^2\delta_\eps<1 \qquad\mbox{and}\qquad \| C_{K,\infty}(\lambda) \|^2\frac{Z^2 \delta_\eps}{1- Z^2 \delta_\eps }<\varepsilon.
	\end{equation*}
	By \partref{ConvTransformed}{eq:uliged2}, 
	there is $\Lambda_\eps>K$ such that
	$ q_{K,\Lambda} := q_{H_{K,\Lambda}}-q_{H_{K,\infty}} $ satisfies $$ |q_{K,\Lambda}(\psi,\psi)| \le \delta_\eps\|H_0^{1/2}\psi\|^2 \qquad\mbox{for all}\ \psi\in\cD(H_0^{1/2})\ \mbox{and}\ \Lambda\ge\Lambda_\eps. $$
	Hence, by \cref{lem:resolventformbound}, the sesquilinear form $q_{K,\Lambda}(R_{K,\infty}(\lambda)^{1/2}\cdot,R_{K,\infty}(\lambda)^{1/2}\cdot)$ is bounded by $Z^2\delta_\eps$, so there is a corresponding operator $D_\eps$ with $\lVert D_\eps\lVert \leq   Z^2 \delta_\eps<1$. Further, we have
	\begin{align*}
	R_{K,\Lambda}(\lambda) = R_{K,\infty}(\lambda)^{1/2}(1+D_{\eps})^{-1}R_{K,\infty}(\lambda)^{1/2}
	\end{align*}
	and hence
	\begin{align*}
	\lVert C_{K,\Lambda}(\lambda)C_{K,\Lambda}(\lambda)^*-C_{K,\infty}(\lambda)C_{K,\infty}(\lambda)^*\lVert&= \lVert C_{K,\infty}(\lambda)((1+D_\eps)^{-1}-1)C_{K,\infty}(\lambda)^*\lVert\\&\le\| C_{K,\infty}(\lambda) \|^2\frac{Z^2\delta_\eps}{1-  Z^2 \delta_\eps }<\varepsilon,
	\end{align*}
	finishing the proof.
\end{proof}
\noindent
We want to prove a similar theorem for the fiber operators. Therefore, let $\sigma\le K<\Lambda<\infty$ and define
\begin{align}\label{defn:Grosstransformedfiber}
&H_{K,\Lambda}(\xi)=W(B_{K,\Lambda})H_{\Lambda}(\xi)W(B_{K,\Lambda})^*+E_{\Lambda} \quad\mbox{and}\\& \Sigma_{K,\Lambda}(\xi)=\inf\sigma (H_{K,\Lambda}(\xi))= \Sigma_\Lambda(\xi)+E_\Lambda.
\end{align}
Let $V$ be the unitary transformation from \cref{Prop:LLP}. Using \cref{Prop:Bregninger} we have
\begin{equation*}
VU_{K,\Lambda}V^*=(F\otimes 1)^* \int_{\IR^\nu}^{\oplus}\Gamma(V_x)W(V_x^*B_{K,\Lambda})\Gamma(V_x^*)dx (F\otimes 1)=\int_{\IR^\nu}^{\oplus}W(B_{K,\Lambda})dx.
\end{equation*}
Combined with \cref{Prop:LLP} this implies
\begin{equation*}
\label{eq:decompositiongrosshamilton}
VH_{K,\Lambda}V^*=\int_{\IR^\nu}^{\oplus}H_{K,\Lambda}(\xi)d\xi.
\end{equation*}
\begin{lem}\label{dom:gross}
	$H_{K,\Lambda}(\xi)$ is selfadjoint on $\cD(H_0(0))$ for all $\xi\in \RR^\nu$.
\end{lem}
\begin{proof}
	Since $\cD(H_{K,\Lambda}(\xi))=W(B_{K,\Lambda})\cD(H_{\Lambda}(\xi))=W(B_{K,\Lambda})\cD(H_{\Lambda}(0))$  for all $\xi$ we see the domain does not depend on $\xi$. As $\cD(H_0)=\cD(H_{K,\Lambda})$, it follows from \cref{lem:dirIntDomain} that $\cD(H_{K,\Lambda}(\xi))=\cD(H_{0}(\xi))=\cD(H_{0}(0))$ for almost all $\xi\in \RR^\nu$, hence it must hold for all $\xi\in \RR^\nu$. 
\end{proof}
\noindent
Throughout this paper, we will need to treat differences of fiber operators appropriately. To this end, for $k,\xi\in \RR^\nu$ and $s\in [0,1]$, we define  the operators
\begin{align}\label{defn:Dxik}
D_\xi(k) &:= 2k\cdot(\xi-d\Gamma(m))+|k|^2=-2d\Gamma(k\cdot m)+2k\cdot \xi+|k|^2,\\ B_\xi &:= (H_0(\xi)+1)^{1/2}. \label{defn:Bxi}\\ A_s &:= 1+d\Gamma(\omega)^{1/2} + |d\Gamma(m)|^s  \label{defn:As} 
\end{align}
Note that
\begin{equation}\label{eq:Dxikdiff}
H_\Lambda(\xi+k) = H_\Lambda(\xi) + D_\xi(k) \qquad\mbox{for all}\ \Lambda\in[0,\infty),\ \xi,k\in\IR^\nu,
\end{equation}
by \cref{eq:fiber-transformation}. We collect the following statements in a \lcnamecref{lem:AohBxi2} for future convenience.
\begin{lem}\label{lem:AohBxi2}
	$A_s$ is selfadjoint on $\cD(A_s)=\cD(|d\Gamma(m)|^s)\cap \cD(d\Gamma(\omega)^{1/2})$. Further, $\cD(B_\xi)\subset \cD(A_s)$, $\cD(A_1)=\cD(B_\xi)$ and $\cD(B_\xi)=\cQ(H_0(\xi))=\cQ(H_0(0))=\cD(B_0)$ for all $s\in[0,1]$ and $\xi\in \RR^\nu$.
\end{lem}
\begin{proof}
	Selfadjointness and the domain follows from the fact that $A_s$ is a sum of commuting, non-negative selfadoint operators, while $\cD(B_\xi)=\cQ(H_0(\xi))=\cQ(H_0(0))=\cD(B_0)$ follows from by \cref{prop:operatorsareSA} and \cref{lem:domainsqrt}.
	Using \cref{prop:operatorsareSA} and \cref{lem:domainsqrt} with $A=H_0(0)$ and $B\in \{|d\Gamma(m)|^{2s}, d\Gamma(\omega)\}$, we see $\cD(B_0)= \cQ(H_0(0))\subset \cD(A_s)$ for all $s\in[0,1]$ and $\xi\in \RR^\nu$. Subadditivity of the square root also shows $B_0 A_1^{-1}$ acts on $\cF^{(n)}$ like multiplication by a function bounded uniformly by 1, so $\cD(A_1)\subset \cD(B_0)=\cD(B_\xi)$.
\end{proof}
\noindent
We will also need the following two technical \lcnamecrefs{lem:Dxik}.
\begin{lem}\label{lem:Dxik}
	Let $a\geq 0$ and $c\in \RR$. The operator $\lvert D_\xi(k)-c \lvert^{a}$ is $B_\xi^a$-bounded for all $\xi,k\in\IR^\nu$.
	Further, $\|D_\xi(k)B_\xi^{-1}\|\le 4|k|$ holds for all $k,\xi\in\IR^\nu$ with $|k|\le1$ and
	\begin{align}\label{eq:11}
	&B_{\xi+k}^{-a} = B_\xi^{-a}(1+D_\xi(k)B_\xi^{-2})^{-a/2} \qquad\mbox{for $|k|<\frac{1}{4}$ }\\
	& \lim_{k\to 0} (1+D_\xi(k)B_\xi^{-2})^{-a/2} = 1.
	\end{align}
\end{lem}
\begin{proof}
	We observe the operator $\lvert D_\xi(k)-c\lvert^a B_\xi^{-a}$ acts on $\FS^{(n)}$ as multiplication by the function
	$$ f_n(k_1,\ldots,k_n) = \frac{\lvert  2k\cdot(\xi-k_1-\cdots-k_n)+|k|^2-c \lvert^a }{(1+\omega^{(n)}(k_1,\ldots,k_n)+|\xi-k_1-\cdots-k_n|^2)^a}. $$
	Since $|f_n|\le C_a(|k|^a+||k|^2-c|^{a})$ with $C_a=2^a$ for $a\in[0,1]$ and $C_a=4^a/2$ for $a>1$, the operator $\lvert D_\xi(k)-c\lvert^aB_\xi^a$ is bounded. Hence, $\lvert D_\xi(k)\lvert^a$ is $B_\xi^a$-bounded and the bound $\|D_\xi(k)B_\xi^{-1}\|\le 4|k|$ holds for $\lvert k\lvert <1$.
	For $|k|<\frac{1}{4}$ we have $\|D_\xi(k)B_\xi^{-2}\|<4|k|<1$ as $\|B_\xi^{-1}\|\leq 1$ and hence $1+D_\xi(k)B_\xi^{-2}$ is invertible. \cref{eq:11} follows using \cref{eq:Dxikdiff} and the fact that $H_0(\xi)$ and $D_\xi(k)$ commute. We now see
	\begin{equation*}
	\|(1+D_\xi(k)(H_0(\xi)+1)^{-1})^{-a/2}-1\| \le \sup_{|x|\le|k|}|(1+4x)^{-a/2}-1| \xrightarrow{k\rightarrow 0} 0.
	\qedhere
	\end{equation*}
\end{proof}
\begin{lem}\label{lem:AohBxi}
	Let $\xi\in\IR^\nu$. Then $B_\xi^{1/2}$ is infinitesimally $A_{\frac 23}$ bounded.
\end{lem}
\begin{proof}
	Let $\eps>0$. Pick $C$ such that
	$$ x^{1/4}\le \frac 12\eps x^{1/2} + C \qquad\mbox{and}\qquad x^{1/2}\le \frac 1{2^{5/4}}\eps x^{2/3} +C \qquad\mbox{for all}\ x\ge 0. $$
	Then the sub-additivity of $x\mapsto x^{1/4}$ and $\lvert \xi-k_1-\cdots-k_n\lvert^2\leq 2\lvert \xi\lvert^2+2\lvert k_1+\cdots+k_n\lvert^2$ lead to
	\begin{align*}
	&\frac{(1+\omega^{(n)}(k_1,\ldots,k_n)+\lvert \xi-k_1-\cdots-k_n\lvert^2 )^{1/4}}{1+\omega^{(n)}(k_1,\ldots,k_n)^{1/2}+\lvert k_1+\cdots+k_n\lvert^{2/3}}\\&\qquad\le\varepsilon + (1+2^{1/4}\lvert \xi\lvert^{1/2}+3C)(1+\omega^{(n)}(k_1,\ldots,k_n)^{1/2}+\lvert k_1+\cdots+k_n\lvert^{2/3})^{-1}.
	\end{align*}
	As this holds uniformly for all $n\in\IN$, we obtain $\cD(A_{\frac{2}{3}}) \subset \cD(B_\xi^{1/2})$.
	To finish the proof we observe
	\begin{equation*}
	\lVert  B_\xi^{1/2} \psi \lVert\leq \varepsilon \lVert A_{\frac{2}{3}}\psi \lVert+(1+2^{1/4}\lvert \xi\lvert^{1/2}+3C)\lVert \psi\lVert
	\qquad\mbox{for all}\ \psi\in \cD(A_{\frac{2}{3}}).
	\qedhere
	\end{equation*}
\end{proof}
\noindent
The next \lcnamecref{lem:convOnePoint} will allow us to reduce the proof of convergence for all $\xi \in \RR^\nu$ to the proof of convergence at a single value $\xi_0\in\IR^\nu$, both for the transformed ($a=1$) and the non-transformed ($a=0$) fiber operators.
\newcommand{\HaKL}[1]{\widetilde{H}_{#1}}\newcommand{\RaKL}[1]{\widetilde R_{#1}}\newcommand{\DaKL}[1]{\widetilde D_{#1}}\newcommand{\CaKL}[1]{C_{#1}}
\newcommand{\caK}{c}
\begin{lem}\label{lem:convOnePoint} Fix $a\in \{0,1\}$ and $K\ge \sigma$. We define
	\[ \HaKL{\Lambda}(\xi) = W(aB_{K,\Lambda})^*H_{\Lambda}(\xi)W(aB_{K,\Lambda})+E_\Lambda \qquad \mbox{for}\ \Lambda\in(K,\infty). \]
	Assume there are $\xi_0\in \RR^\nu$ and $ \lambda_0 \in \RR$ such that
	\begin{enumerate}[-]
		\item $\HaKL{\Lambda}(\xi_0)$ converges to an operator $\HaKL{\infty}(\xi_0)$ in the norm resolvent sense as $\Lambda\to\infty$,
		\item $\lambda_0\leq \HaKL{\Lambda}(\xi_0)-1$ for all $\Lambda\in (K,\infty]$,
		\item $\cQ(\HaKL{\infty}(\xi_0))\subset \cD(A_{2/3})$ and for all $\lambda\leq \lambda_0$
		\begin{equation}\label{eq:funnyConv}
		\slim_{\Lambda\to\infty}A_{2/3}(\HaKL{\Lambda}(\xi_0) - \lambda)^{-1/2}=A_{2/3}(\HaKL{\infty}(\xi_0) - \lambda)^{-1/2}.
		\end{equation}
	\end{enumerate}
	Then $\HaKL{\Lambda}(\xi)$ converges to an operator $\HaKL{\infty}(\xi)$ in the strong resolvent sense for all $\xi\in \RR^\nu$.
	Further, for $\Lambda\in(K,\infty]$ and $\xi\in\IR^\nu$, let \[\DaKL{\Lambda}(\xi)=W(aB_{K,\Lambda})^*D_{\xi_0}(\xi-\xi_0)W(aB_{K,\Lambda}).\]
	Then $q_{\DaKL{\infty}(\xi)}$ is infinitesimally $q_{\HaKL{\infty}(\xi_0)}$ bounded and  
	\begin{equation*}
	q_{\HaKL{\infty}(\xi)} = q_{\HaKL{\infty}(\xi_0)} + q_{\DaKL{\infty}(\xi)}.
	\end{equation*}
\end{lem}
\begin{rem}\label{rem:formdomain}
	We note that $\cQ(\HaKL{\Lambda}(\xi_0))\subset W(aB_{K,\Lambda})^*\cD(A_{2/3})=\cD(A_{2/3})$ for all $\Lambda\in(K,\infty)$, by \cref{lem:propertiesOfB} and \partref{thm:Weyltransformation-strong}{part:tothepowers}. Hence, $A_{2/3}(\HaKL{\Lambda}(\xi_0) - \lambda)^{-1/2}$ is a bounded operator for all $\lambda\leq \lambda_0$ and the left hand side of \cref{eq:funnyConv} is well-defined without further assumptions.
\end{rem}
\begin{proof}
	For $\Lambda\in(K,\infty]$ and $\lambda\leq \lambda_0$, we define $\RaKL{\Lambda}(\lambda)=(\HaKL{\Lambda}(\xi_0) - \lambda)^{-1}$ and 
	$\CaKL{\Lambda}=A_{2/3}W(aB_{K,\Lambda})A_{2/3}^{-1}$. By \cref{lem:propertiesOfB} and \partref{thm:Weyltransformation-strong}{part:tothepowers}, $\CaKL{\Lambda}$ is bounded, $\Lambda\mapsto \CaKL{\Lambda}$ is strongly continuous and $\CaKL{\Lambda}$ strongly converges to $\CaKL{\infty}$ as $\Lambda\to\infty$. By the uniform boundedness principle, there is $\caK>0$ such that
	\begin{equation}\label{eq:CaKLbound}
	\lVert \CaKL{\Lambda}\lVert\le \caK
	\qquad
	\mbox{for all}\ \Lambda\in(K,\infty].
	\end{equation}
	Let $(\Lambda_n)\subset(K,\infty)$ be an arbitrary sequence converging to $\infty$ as $n\to \infty$ and write $\cS=\{ \infty \}\cup \{\Lambda_n | n\in  \NN \}$. 
	Further, we note that $A_{2/3}\RaKL{\Lambda}(\lambda)^{1/2}$ is bounded for all $\lambda\leq \lambda_0$ and $\Lambda\in(K,\infty]$, by \cref{rem:formdomain} and the assumption.
	Now, combing the observation that $\lVert \RaKL{\Lambda}(\lambda_0)^{-1/2}\RaKL{\Lambda}(\lambda)^{1/2}\lVert \leq 1 $ for all $\lambda\leq \lambda_0$, by the spectral theorem, with the uniform boundedness principle and \cref{eq:funnyConv}, we find there exists $b>0$ such that
	\begin{equation}\label{eq:RaKLbound}
	\lVert A_{2/3}\RaKL{\Lambda}(\lambda)^{1/2}\lVert \leq  b
	\qquad
	\mbox{for all}\ 
	\Lambda\in \cS \ \mbox{and}\ \lambda\le\lambda_0.
	\end{equation}
	Let $\varepsilon>0, \xi\in \RR^\nu$ and $h=\xi-\xi_0$. By \cref{lem:Dxik,lem:AohBxi}, there is $a_{\varepsilon}>0$ such that
	\begin{equation}\label{eq:infbound}
	\lVert \lvert D_{\xi_0}(h) \lvert^{1/2} \psi  \lVert \leq \frac{\sqrt{\varepsilon}}{\sqrt{2} b \caK}\lVert A_{2/3} \psi\lVert+a_{\varepsilon}\lVert \psi\lVert  \qquad\mbox{for all}\ \psi\in\cD(A_{2/3}).
	\end{equation}
	For $\Lambda\in (K,\infty]$ and  $\lambda\leq  \lambda_0$, we have
	\begin{equation*}
	\lvert D_{\xi_0}(h) \lvert^{1/2}W(aB_{K,\Lambda})\RaKL{\Lambda}(\lambda)^{1/2}= \lvert D_{\xi_0}(h) \lvert^{1/2}A_{2/3}^{-1} \CaKL{\Lambda} A_{2/3}\RaKL{\Lambda}(\lambda)^{1/2}
	\end{equation*}
	which is bounded and strongly converges to $\lvert D_{\xi_0}(h) \lvert^{1/2}W(aB_{K,\infty})\RaKL{\infty}(\lambda)^{1/2}$ as $\Lambda\to\infty$. Inserting the bounds \cref{eq:infbound,,eq:CaKLbound,,eq:RaKLbound}, for $\Lambda\in \cS$, $\lambda\le\lambda_0$ and $\psi\in\FS$, we obtain
	\begin{align}
	\lVert &\lvert D_{\xi_0}(h) \lvert^{1/2}W(aB_{K,\Lambda})\RaKL{\Lambda}(\lambda)^{1/2} \psi  \lVert^2 \nonumber
	\\&\qquad \le \frac{\varepsilon}{(b \caK)^2}\| \CaKL{\Lambda} A_{2/3}\RaKL{\Lambda}(\lambda)^{1/2} \psi \|^2 + 2a_\eps^2 \|W(aB_{K,\Lambda})\RaKL{\Lambda}(\lambda)^{1/2} \psi\|^2\label{eq:Ybound} \\
	&\qquad  \leq  \varepsilon\lVert \psi \lVert^2 + 2a_{\varepsilon}^2 \lVert \RaKL{\Lambda}(\lambda)^{1/2} \psi  \lVert^2 \nonumber
	\end{align}
	In particular, this implies $\cQ(\HaKL{\infty}(\xi_0))\subset \cQ(\DaKL{\infty}(\xi))$ and
	\begin{equation*}
	\lvert q_{\DaKL{\infty}(\xi)} (\psi,\psi)\lvert \leq \varepsilon q_{\HaKL{\infty}(\xi_0)} (\psi,\psi) + (\lambda_0 +2a_{\varepsilon}^2 \lVert \RaKL{\Lambda}(\lambda_0)^{1/2}\lVert^2)  \lVert \psi \lVert^2 \quad\mbox{for}\ \psi\in\cQ(\HaKL{\infty}(\xi_0)).
	\end{equation*}
	Hence, the form of $\DaKL{\infty}(\xi)$ is infinitesimally $\HaKL{\infty}(\xi_0)$-form bounded.
	
	We define $\HaKL{\infty}(\xi)$ as the selfadjoint operator corresponding to the form $q_{\HaKL{\infty}(\xi_0)} + q_{\DaKL{\infty}(\xi)}$, by the KLMN theorem. Further, for $\Lambda\in (K,\infty]$ and $\lambda\leq \lambda_0$, let
	\begin{align*}
	& Y_{\Lambda, \lambda} = \lvert D_{\xi_0}(h) \lvert^{1/2}W(aB_{K,\Lambda})\RaKL{\Lambda}(\lambda)^{1/2}\\
	& Z_{\Lambda, \lambda} = Y_{\Lambda, \lambda}^* \operatorname{Sign}(D_{\xi_0}(h)) Y_{\Lambda, \lambda}.
	\end{align*}
	We pick $\lambda_1<\lambda_0$ such that $2a_{\frac{1}{4}}^2(\lambda_1-\lambda_0)^{-1}\le \frac{1}{4}$ and from now assume $\lambda\leq \lambda_1$. By the spectral theorem and the bound \cref{eq:Ybound}, we find $\|Y_{\Lambda,\lambda}\|\le \frac{1}{\sqrt{2}}$ and $\|Z_{\Lambda,\lambda}\|\le \frac 12$ for all $\Lambda\in \cS$. Now, $Y_{\Lambda_n, \lambda}$ strongly converges   to $Y_{\infty, \lambda}$ as $n\to\infty$ by the previous arguments. Further, by \cref{WeakStrongConvergence}, we see $Y_{\Lambda_n, \lambda}^*$ strongly converges  to $Y_{\infty, \lambda}^*$ as $n\to\infty$. Hence, $Z_{\Lambda_n, \lambda} $ strongly converges to $Z_{\infty, \lambda}$ as $n\to\infty$. Furthermore, we see
	\begin{equation*}
	\HaKL{\Lambda_n}(\xi) - \lambda = \RaKL{\Lambda_n}(\lambda)^{-1/2}(1+Z_{\Lambda_n, \lambda})\RaKL{\Lambda_n}(\lambda)^{-1/2},
	\end{equation*}
	so $\HaKL{\Lambda_n}(\xi) - \lambda$ is invertible and 
	\begin{equation*}
	(\HaKL{\Lambda_n}(\xi) - \lambda)^{-1} = \RaKL{\Lambda_n}(\lambda)^{1/2}(1+Z_{\Lambda_n, \lambda})^{-1}\RaKL{\Lambda_n}(\lambda)^{1/2}.
	\end{equation*}
	By \cref{lem:resolventformbound} and the definition of $\HaKL{\infty}(\xi)$, we also have that $\HaKL{\infty}(\xi) - \lambda$ is invertible with
	\begin{equation*}
	(\HaKL{\infty}(\xi) - \lambda)^{-1} = \RaKL{\infty}(\lambda)^{1/2}(1+Z_{\infty, \lambda})^{-1}\RaKL{\infty}(\lambda)^{1/2}.
	\end{equation*}
	By the dominated convergence theorem, we now see
	\begin{equation*}
	\slim_{n\to\infty} (1+Z_{\Lambda_n, \lambda})^{-1} =  \sum_{n=0}^{\infty}\slim_{n\rightarrow\infty} (-Z_{\Lambda_n, \lambda})^n = (1+Z_{\infty, \lambda})^{-1}.
	\end{equation*}
	Hence, $(-\infty, \lambda_1]$ is contained in the resolvent set of $\HaKL{\Lambda}(\xi)$ for all $\Lambda\in\cS$ and $(\HaKL{\Lambda_n}(\xi) - \lambda)^{-1}$ strongly converges  to $(\HaKL{\infty}(\xi) - \lambda)^{-1}$ for all $\lambda\leq \lambda_1$. The strong convergence of $(\HaKL{\Lambda_n} - z_0)^{-1}$ to $(\HaKL{\infty}(\xi) - z_0)^{-1}$ at some point $z_0 \in \CC\setminus \RR$ simply follows by power series expansion of $z\mapsto (\HaKL{\Lambda}(\xi) - z)^{-1}$ around $z=\lambda_1-1$ for $\Lambda\in \cS$. Since the sequence $(\Lambda_n)$ with $\lim_{n\to\infty}\Lambda_n=\infty$ was arbitrary and the limit $\HaKL{\infty}(\xi)$ is independent of its choice, this finishes the proof of strong resolvent convergence. 
\end{proof}
\noindent
Similar to \cref{defn:resolvent-HKL}, we define the fiber resolvents
\begin{equation}\label{defn:fiber-resolvent}
R_{\xi,\Lambda}(z) = (H_\Lambda(\xi)+E_\Lambda-z)^{-1} \quad\mbox{and}\quad R_{\xi,K,\Lambda}(z)=(H_{K,\Lambda}(\xi)-z)^{-1},
\end{equation}
which at this point is defined for $\sigma\le K\le\Lambda< \infty$. Starting with \cref{transformedFib}, we will also treat the case $\Lambda=\infty$, setting $E_\infty = 0$.

To renormalize fiber operators, we will need a collection of continuity statements for these resolvents. They are collected in the next \lcnamecref{Lem:continuity}.
\begin{prop}\label{Lem:continuity}
	The following holds.
	\begin{enumerate}[(1)]
		\item\label{part:nonTrans} Fix $\Lambda\in[0,\infty)$. Then $H_\Lambda(\xi)$ is uniformly bounded below in $\xi\in\IR^\nu$ and for $\lambda<\inf_{\xi\in\IR^\nu}(\Sigma_\Lambda(\xi)+E_\Lambda)$ the map $\xi\mapsto R_{\xi,\Lambda}(\lambda)$ is continuous on $\IR^\nu$.
		
		\item\label{part:GrossTrans}  Fix $\sigma\le K<\Lambda<\infty$. Then $H_{K,\Lambda}(\xi)$ is uniformly bounded below in $\xi\in\IR^\nu$ and $ \cQ(H_{K,\Lambda}(\xi))\subset \cD(B_\xi)$ for all $\xi\in\IR^\nu$. If $\lambda<\inf_{\xi\in\IR^\nu}\Sigma_{K,\Lambda}(\xi)$ and $a\in [0,1]$ then the map
		$$ r_{a,\lambda}: \xi \mapsto B_\xi^aR_{\xi,K,\Lambda}(\lambda)^{1/2}(B_\xi^aR_{\xi,K,\Lambda}(\lambda)^{1/2})^* $$
		is continuous in norm.
	\end{enumerate}
\end{prop}
\begin{proof}
	To prove \cref{part:nonTrans}, we first note the uniform lower bound follows directly from \partref{Prop:Bregninger}{part:secondquantizedupperbounds}. 
	Now fix $\Lambda, \lambda$ as stated in \cref{part:nonTrans} and $\xi\in \RR^\nu$. From \cref{prop:operatorsareSA,lem:Dxik,lem:domainsqrt} we know $\cQ({H_\Lambda(\xi)})=\cQ(H_0(\xi))=\cD(B_\xi)\subset \cD(D_\xi(k))$ and
	$$ \|D_\xi(k)R_{\xi,\Lambda}(\lambda)^{1/2}\| \leq  \|D_\xi(k)B_\xi^{-1}\|\|B_\xi R_{\xi,\Lambda}(\lambda)^{1/2}\| \xrightarrow{k\to 0} 0.  $$
	Recalling \cref{eq:Dxikdiff,lem:Dxik} we obtain
	$$ R_{\xi+k,\Lambda}(\lambda) = R_{\xi,\Lambda}(\lambda)(1+D_\xi(k)R_{\xi,\Lambda}(\lambda))^{-1} $$
	for $|k|$ sufficiently small. Taking the limit $k\to 0$ proves the claim.
	
	The uniform lower bound in \cref{part:GrossTrans} follows directly from \cref{part:nonTrans,defn:Grosstransformedfiber}, while $\cQ(H_{K,\Lambda}(\xi))\subset \cD(B_\xi)$ follows from \cref{dom:gross,lem:domainsqrt}. Fix $K,\Lambda,\lambda$ as stated in \cref{part:GrossTrans} and $\xi\in \RR^\nu$. Define
	$\widetilde{D}_\xi (k)=W(B_{K,\Lambda})D_{\xi}(k)W({B_{K,\Lambda}})^*$ and note that
	$$\widetilde{D}(k)R_{\xi,K,\Lambda}(\lambda)^{1/2}=W(B_{K,\Lambda}) D_\xi(k)R_{\xi,\Lambda}(\lambda+E_\Lambda)^{1/2} W(B_{K,\Lambda})^*\xrightarrow{k\to 0} 0
	$$
	in norm. In particular, 
	$$ Z(k)=R_{\xi,K,\Lambda}(\lambda)^{1/2}\widetilde{D}(k)R_{\xi,K,\Lambda}(\lambda)^{1/2} $$ is bounded and goes to 0 for $k\rightarrow 0$. We easily deduce $H_{K,\Lambda}(\xi+h)=H_{K,\Lambda}(\xi)+\widetilde{D}(h)$ from \cref{eq:Dxikdiff} and therefore obtain
	\begin{equation*}
	R_{\xi+k,K,\Lambda}(\lambda) = R_{\xi,K,\Lambda}(\lambda)^{1/2}(1+Z(k))^{-1}R_{\xi,K,\Lambda}(\lambda)^{1/2}
	\end{equation*}
	for sufficiently small $k$. Setting $C=B_\xi^aR_{\xi,K,\Lambda}(\lambda)^{1/2}$ this yields
	$$
	r_{a,\lambda}(\xi+k)=(1+D_\xi(k)B_\xi^{-2})^{a/2} C(1+{Z}(k))^{-1} C^*(1+D_\xi(k)B_\xi^{-2})^{a/2},
	$$
	so $r_{a,\lambda}(\xi+k)$ converges to $CC^*=r_a(\xi)$ in norm as $k\to 0$ by \cref{lem:Dxik}.
\end{proof}
\noindent
We can now renormalize the transformed fiber operators.
\begin{thm}\label{transformedFib}
	Let $K\ge \sigma$ and $V$ be the transformation from \cref{Prop:LLP}.\\ Then the following holds:
	\begin{enumerate}[(1)]\label{ConvGross}
		\item\label{ConvGrossMain} The operators $H_{K,\Lambda}(\xi)$ are bounded below uniformly in $\Lambda$ and $\xi$. Further, $H_{K,\Lambda}(\xi)$ converges to an operator $H_{K,\infty}(\xi)$ in the norm resolvent sense as $\Lambda\to\infty$ for all $\xi\in \RR^\nu$. Convergence is uniform in $\xi$, so the map $\xi\mapsto H_{K,\infty}(\xi)$ is continuous in the norm resolvent sense and
		$$ V H_{K,\infty} V^* =\int_{\IR^\nu}^{\oplus}H_{K,\infty}(\xi)d\xi. $$
		\item \label{dom} $\cQ({H_{K,\infty}(\xi)})\subset\cD(B_\xi)=\cQ({H_0(0)})$ for all $\xi \in \RR^\nu$ and for all $\lambda<\Sigma_{K,\infty}(\xi)$ we have
		\begin{align*}
		\slim_{\Lambda\to\infty}B_\xi R_{\xi,K,\Lambda}(\lambda)^{1/2} = B_\xi R_{\xi,K,\infty}(\lambda)^{1/2}=:h_\lambda(\xi).
		\end{align*}
		Further, the map $\xi\mapsto h_\lambda(\xi)h_\lambda(\xi)^* $ is continuous.
		\item \label{domeq} $\cQ({H_{K,\infty}(\xi)})=\cQ({H_0(0)})$ for all $\xi \in \RR^\nu$.
	\end{enumerate}
\end{thm}
\begin{proof}
	We will first prove \cref{ConvGrossMain,dom} up to a nullset $\cN\subset \IR^\nu$. Then, we show that $\cN$ can be taken to be empty.
	
	From \partref{Lem:continuity}{part:GrossTrans} with $a=0$, we see $\{ H_{K,\Lambda}(\xi)\}_{\xi \in\RR^\nu}$ is uniformly lower bounded and continuous in the norm resolvent sense. Hence, by \cref{Lem:ConvergenceUniform} and \partref{ConvFullOp}{part:Grosstransformedfullconvergence}, we see that $H_{K,\Lambda}(\xi)$ is uniformly bounded below in $\Lambda$  and there is a nullset $\cN$ such that $H_{K,\Lambda}(\xi)$ converges to an operator $H_{K,\infty}(\xi)$ in the norm resolvent sense as $\Lambda\to\infty$ for all $\xi\in\cN^c$ and
	\begin{equation*}
	\cN \subset  \{ \xi\in \RR^\nu \mid \textup{$H_{K,\Lambda}(\xi)$ does not converge in strong resolvent sense as $\Lambda\rightarrow \infty$} \}.
	\end{equation*}
	Pick $\lambda\in\IR$ such that $\lambda<\Sigma_{K,\Lambda}(\xi)$ for all $\Lambda\in (K,\infty)$ and $\xi\in \RR^\nu$. For finite $\Lambda>K$ we use \cref{Prop:Directint,lem:dirIntDomain} to obtain $h_\Lambda(\xi)= B_\xi R_{\xi,K,\Lambda}(\lambda)^{1/2}$ is essentially bounded and
	\begin{align*}
	V(H_0+1)^{1/2} R_{K,\Lambda}(\lambda)^{1/2}((H_0+1)^{1/2} R_{K,\Lambda}(\lambda)^{1/2})^* V^*= \int_{\IR^\nu}^{\oplus}h_\Lambda(\xi)h_\Lambda(\xi)^* d\xi.
	\end{align*}
	The left hand side converges in norm as $\Lambda\to\infty$ by \partref{ConvFullOp}{part:GrossFullSqrtConv}.
	Hence, \cref{Lem:conccc} and \partref{Lem:continuity}{part:GrossTrans} with $a=1$ yield $h_\Lambda(\xi)h_\Lambda(\xi)^*$ converges uniformly to a continuous function as $\Lambda\to\infty$. By \cref{Thm:ConvRelBound}, this implies that \cref{dom} holds on $\cN^c$. Combining \cref{lem:AohBxi2} and \cref{lem:convOnePoint}, we now see $\cN=\emptyset$ and that $\cQ({H_{K,\infty}(\xi)})$ is independent of $\xi$. This proves \cref{ConvGrossMain} and $\cref{dom}$.
	
	To prove \cref{domeq}, note $\cQ(H_{K,\infty})=\cQ(H_0)$ \partref{ConvFullOp}{part:Grosstransformedfullconvergence}, so $\cQ(H_{K,\infty}(\xi))=\cQ(H_0(\xi))$ for almost all $\xi\in \RR^\nu$ by \cref{lem:dirIntDomain}. As $\cQ(H_0(\xi))=\cQ(H_0(0))$ by \cref{prop:operatorsareSA} and \cref{lem:domainsqrt}, the conclusion follows since $\cQ({H_{K,\infty}(\xi)})$ is independent of $\xi$.
\end{proof}
\noindent
Finally, we go back to the non-transformed fiber operators and finish the renormalization in the next \lcnamecref{LargeFiberProp}. Especially, we establish a connection between the regularity of $W(B_{K,\infty})^*=W(-B_{K,\infty})$ and the regularity of the domain of $H_\infty(\xi)$. 
\begin{thm}\label{LargeFiberProp}
	Let $s\in [0,1]$. Then the following holds:
	\begin{enumerate}[(1)]
		\item\label{part:FiberConv} $H_{\Lambda}(\xi)+E_\Lambda$ is uniformly bounded below in $\xi$ and $\Lambda$ and converges to an operator $H_\infty(\xi)$ in the norm resolvent sense as $\Lambda\to\infty$. Further, $\xi \mapsto H_\infty(\xi)$ is norm resolvent continuous and
		$$
		VH_\infty V^*=\int_{\IR^\nu}^{\oplus}H_\infty(\xi)d\xi.
		$$
		\item \label{Formula}
		For all $K\ge\sigma$ we have $H_\infty(\xi)=W(B_{K,\infty})^* H_{K,\infty}(\xi) W(B_{K,\infty}).$
		Especially,  $\cQ({H_\infty(\xi)})\subset W(B_{K,\infty})^* \cQ({H_{0}(0)})$ holds for all $\xi\in\IR^\nu$. Further, for every $K>\sigma$ we have $\cQ({H_\infty(\xi)})= W(B_{K,\infty})^* \cQ({H_{0}(0)})$ all $\xi\in\IR^\nu$.
		\item\label{part:FiberConvSqRt} 
		Fix $K\ge\sigma$. If $B_{K,\infty}\in \cD(\lvert m \lvert^s)$, then $\cD(H_\infty(\xi)) \subset \cQ({H_\infty}(\xi))\subset \cD(A_s)$ and
		\begin{align*}\label{eq:Convergence}
		\slim_{\Lambda\rightarrow \infty}  A_sR_{\xi,\Lambda}(\lambda)^{1/2} = A_sR_{\xi,\infty}(\lambda)^{1/2}=:g_\lambda(\xi) \quad\mbox{for any}\ \lambda<\Sigma_\infty(\xi).
		\end{align*}
		Further, for $\lambda$ sufficiently small the maps $\xi\mapsto g_\lambda(\xi)g_\lambda(\xi)^*$ and $\xi\mapsto \lVert g_\lambda(\xi)\lVert$ are continuous.
	\end{enumerate} 
\end{thm}
\begin{proof}
	First, for any $\xi\in\IR^\nu$, we can conclude from \partref{transformedFib}{ConvGrossMain}  that
	\begin{align*}
	\slim_{\Lambda\rightarrow \infty} (H_{\Lambda}(\xi)+E_\Lambda+i)^{-1}&=\slim_{\Lambda\rightarrow \infty} (W(B_{\Lambda,K})^*H_{K,\Lambda}(\xi)W(B_{\Lambda,K})+i)^{-1}\\&=(W(B_{\infty,K})^*H_{K,\infty}(\xi)W(B_{\infty,K})+i)^{-1}.
	\end{align*}
	Hence, $H_{\Lambda}(\xi)+E_\Lambda$ converges to $W(B_{\Lambda,K})^*H_{K,\infty}(\xi)W(B_{\Lambda,K})$ in strong resolvent sense. From \partref{Lem:continuity}{part:nonTrans}, we have $\{ H_{\Lambda}(\xi)\}_{\RR^\nu}$ is uniformly lower bounded and continuous in the norm resolvent sense. Hence, we can use \cref{Lem:ConvergenceUniform} and \partref{ConvFullOp}{part:Fulloperatorconv} to see that the convergence is in the norm resolvent sense.
	
	It remains to prove \cref{part:FiberConvSqRt}. The domain statement follows directly from \cref{Formula}, \cref{lem:propertiesOfB} and \partref{thm:Weyltransformation-strong}{part:tothepowers}. For $\Lambda$ large enough that $\lambda<\Sigma_\Lambda(\xi)+E_\Lambda$ (cf. \cref{lem:normresolventConv}), the definition \cref{defn:Grosstransformedfiber} yields
	\begin{align*}
	A_sR_{\xi,\Lambda}(\lambda)^{1/2} = A_sW(B_{K,\Lambda})^*A_s^{-1}A_sB_\xi^{-1} B_\xi R_{\xi,K,\Lambda}(\lambda)^{1/2} W(B_{K,\Lambda}),
	\end{align*}
	for any $K\ge \sigma$. Convergence now follows from \cref{Formula}, \cref{thm:Weyltransformation-strong} and \partref{transformedFib}{dom}. To prove the continuity statements, we first observe $\lVert g_\lambda(\xi)g_\lambda(\xi)^* \lVert=\lVert g_\lambda(\xi)\lVert^2$ so it is enough to see that  $\xi \mapsto g_\lambda(\xi)g_\lambda(\xi)^*$ is continuous in norm. Writing $h_\lambda$ as in \partref{transformedFib}{dom} we have
	$$ g_\lambda(\xi)g_\lambda(\xi)^* = A_sW(B_{K,\infty})^* B_\xi^{-1} h_\lambda(\xi)h_\lambda(\xi)^* (A_sW(B_{K,\infty})^* B_\xi^{-1})^* \quad\mbox{for any}\ K\ge \sigma.$$
	By \partref{thm:Weyltransformation-strong}{part:tothepowers} and \cref{lem:propertiesOfB,lem:Dxik} the map $ \xi\mapsto A_sW(B_{K,\infty})^*B_\xi^{-1} $ is continuous in norm, so the statement follows from \partref{transformedFib}{dom}.
\end{proof}
\noindent
We are now ready to prove \cref{thm:operatordomain}. 
%
\begin{proof}[\textbf{Proof of \cref{thm:operatordomain}}]
	\cref{part:quadraticdomain} follows directly, by combining \cref{LargeFiberProp} with \cref{lem:convOnePoint}. Hence, we only need to prove \cref{part:operatordomain}.
	
	Set $\cD_s=\cD(A_s)=\cD(d\Gamma(\omega)^{1/2})\cap\cD(|d\Gamma(m)|^s)$ and note we have
	$$ \cQ({H_\infty(\xi_1)})=\cQ({H_\infty(\xi_2)}) = W(B_{K,\infty})^*\cD_1,$$
	by \partref{thm:operatordomain}{part:quadraticdomain}.
	Furthermore, \cref{defn:Dxik} and \cref{lem:Dxik,lem:AohBxi2} imply $$ \cD_1=\cD(B_0)=\cD(B_\xi)\subset \cD(D_{\xi_1}(\xi_2-\xi_1))=\cD( d\Gamma(  2 (\xi_2-\xi_1)\cdot m  ). $$ Now assume $(\xi_2-\xi_1)\cdot mB_{K,\infty}\in \HS$. Then $W(B_{K,\infty})^*\cD_1 \subset \cD(D_{\xi_1}(\xi_2-\xi_1))$ by \cref{thm:Weyltransformation-strong}. Hence, \partref{thm:operatordomain}{part:quadraticdomain} yields
	$$ q_{H_{\infty}(\xi_2)}(\psi,\phi) = q_{H_\infty(\xi_1)}(\psi,\phi) + \braket{D_{\xi_1}(\xi_2-\xi_1)\psi,\phi} \quad\mbox{for all}\ \psi,\phi\in W(B_{K,\infty})^*\cD_1. $$
	If we fix $\psi\in W(B_{K,\infty})^*\cD_1$, the map $\phi\mapsto q_{H_\infty(\xi_1)}(\psi,\phi)$ is continuous if and only if the map $\phi\mapsto q_{H_\infty(\xi_2)}(\psi,\phi)$ is continuous. Hence, \cref{lem:operatorandformdomain} yields $\cD(H_\infty(\xi_1))=\cD(H_\infty(\xi_2))$ and $H_\infty(\xi_2)=H_\infty(\xi_1)+D_{\xi_1}(\xi_2-\xi_1)$.  
	
	We move to the case $(\xi_2-\xi_1)\cdot mB_{K,\infty}\notin \HS$ and let $\psi\in \cD(H_\infty(\xi_1))\cap \cD(H_\infty(\xi_2))$. Then it follows from \cref{lem:operatorandformdomain} and \partref{thm:operatordomain}{part:quadraticdomain} that $$W(B_{K,\infty})^* \cD_1 \ni \phi\mapsto q_{D_{\xi_1}(\xi_2-\xi_1)}(\psi,\phi)$$ is continuous. If we can prove $W(B_{K,\infty})^*\cD_1$ is a form core for $d\Gamma((\xi_2-\xi_1)m)$ we can deduce $\psi\in \cD(D_{\xi_1}(\xi_2-\xi_1))$ by \cref{lem:operatorandformdomain}.
	From \cref{thm:Weyltransformation-strong} we then find
	$$ \psi\in \cD(d\Gamma((\xi_2-\xi_1)m))\cap W(B_{K,\infty})^*\cD_1 = \{0\}, $$
	which proves the statement.
	
	As $A_{\frac23}$ dominates $\lvert D_{\xi_1}(\xi_2-\xi_1) \lvert^{1/2}$ by \cref{lem:AohBxi} and $A_{\frac23}$ commutes with $\lvert D_{\xi_1}(\xi_2-\xi_1) \lvert^{1/2}$ we see that any core for $A_{\frac23}$ is a core for $\lvert D_{\xi_1}(\xi_2-\xi_1) \lvert^{1/2}$ by \cref{commuteingCore}. Now we note $A_{\frac 23}$ commutes with $A_1$ and is $A_1$-bounded, since $\cD_1\subset \cD_{\frac 23}$, so $\cD_1$ is a core for $A_{\frac 23}$ by \cref{commuteingCore}. Further, by \cref{thm:Weyltransformation-strong,lem:propertiesOfB} we know $W(B_{K,\infty})^*$ maps $\cD_{\frac 23}$ continuously onto $\cD_{\frac 23}$, so $W(B_{K,\infty}^*)\cD_1$ is a core for $A_{\frac23}$ and hence for $\lvert D_{\xi_1}(\xi_2-\xi_1) \lvert^{1/2}$. This finishes the proof.
\end{proof}

\section{Proof of Ground State Absence}
\label{sec:proof}

In this section we prove \cref{thm:nogs}. We will work under \cref{hyp1,,hyp2,,hyp3} in this section and fix $\nu\ge2$. We will also use some notation from the previous section. Specifically, for $\xi,k\in \RR^\nu$ we define the operators $B_\xi$ and $D_\xi(k)$ as in \cref{defn:Bxi,defn:Dxik}, respectively. Then from \partref{thm:operatordomain}{part:quadraticdomain} and \cref{lem:Dxik,lem:AohBxi} we know
\begin{equation}\label{eq:domains}
\cD(H_\infty(\xi_1) )\subset \cQ({H_\infty(\xi_1)}) \subset \cD( B_{\xi_2}^{1/2} )\subset \cQ({d\Gamma(\xi_3\cdot m)})\quad\mbox{for}\ \xi_1,\xi_2,\xi_3\in \RR^\nu.
\end{equation}
We will work with the vector valued form
\begin{equation}
q_{d\Gamma(m)}(\phi,\psi) = (q_{d\Gamma(m_1)}(\phi,\psi),...,q_{d\Gamma(m_\nu)}(\phi,\psi))
\end{equation}
defined for  $\phi,\psi\in \cD(q_{d\Gamma(m)}):=\bigcap_{i=1}^\nu\cD(\lvert d\Gamma(m_i) \lvert^{1/2}) = \cD(\lvert d\Gamma(m) \lvert^{1/2}) \supset \cD( B_{0}^{1/2} )$ by \cref{Lem:propertiesOfFuncDGamma,lem:domainsqrt}. The following holds 
\begin{lem}\label{lem:qDgammaM}
	Let $k\in \RR^\nu$ and $\psi,\phi\in \cQ(B_0^{1/2})\subset \cQ({d\Gamma(k\cdot m)})$. Then we have $k\cdot q_{d\Gamma(m)}(\phi,\psi)=q_{d\Gamma(k\cdot m)}(\phi,\psi)$.
\end{lem}
\begin{proof}
	This follows directly from
	\[
	q_{d\Gamma(m)}(\phi,\psi) =  \sum_{n=1}^{\infty} \int_{\RR^{n\nu} }(k_1+...+k_n) \overline{\phi^{(n)}(k_1,...,k_n)}\psi^{(n)}(k_1,...,k_n)dk_1...dk_n.\qedhere
	\]
\end{proof}
\noindent
We now observe some direct consequences of the rotation invariance in \cref{hyp3}. The following lemma is well-known, but is so essential to our method that we provide a short proof.
\begin{lem}\label{lem:rotation-operator}
	Let $O\in\IR^{\nu\times\nu}$ be orthogonal and $U_O$ be the associated rotation o\-pe\-ra\-tor acting on $f\in\HS$ as
	$ U_O f (k) =  f(Ok)$. Then $ \Gamma(U_O)^*H_\Lambda(\xi)\Gamma(U_O) = H_\Lambda(O\xi)$ holds for all $\Lambda\in[0,\infty]$ and $\xi\in\IR^\nu$. Further, $\Gamma(U_O)\cQ({H_\Lambda(\xi)})=\cQ({H_\Lambda(\xi)})$.
\end{lem}
\begin{proof}
	Let $\Lambda<\infty$. By \cref{Prop:Bregninger} and rotation invariance of $v$ and $\omega$ we find
	\begin{align*}
	\Gamma(U_O)^*d\Gamma(\omega)\Gamma(U_O) &= d\Gamma(U_O^*\omega U_O )=d\Gamma(\omega),\\
	\Gamma(U_O)^*d\Gamma(\xi\cdot m)\Gamma(U_O) &= d\Gamma(U_O^*(\xi\cdot m) U_O )=d\Gamma(O\xi\cdot m),\\
	\Gamma(U_O)^*\ph(v_\Lambda)\Gamma(U_O) &= \ph(U_O^*v_\Lambda) = \ph(v_\Lambda).
	\end{align*}
	Furthermore, we see
	\begin{equation*}
	(U_O^{\otimes n})^* \lvert k_1+...+k_n \lvert^2 U_O^{\otimes n}=\lvert Ok_1+...+Ok_n \lvert^2=\lvert k_1+...+k_n \lvert^2
	\quad\mbox{for}\ n\in\IN,
	\end{equation*}
	which proves $\Gamma(U_O)^*\lvert d\Gamma(m)\lvert^2\Gamma(U_O)=\lvert d\Gamma(m)\lvert^2$. Hence, by the definition \cref{defn:fiber}, the equality
	$ \Gamma(U_O)^*H_\Lambda(\xi)\Gamma(U_O) = H_\Lambda(O\xi)$ holds for all finite $\Lambda$. Taking the norm resolvent limit yields the case $\Lambda=\infty$. This now implies $\Gamma(U_O)\cQ({H_\Lambda(\xi)})=\cQ({H_\Lambda(O\xi)})=\cQ({H_\Lambda(\xi)})$ by \partref{thm:operatordomain}{part:quadraticdomain}.
\end{proof}
\noindent
Let $\eps\in(0,1)$ and define
\begin{equation}
S_\eps(\xi) = \{k\in\IR^\nu\setminus\{0\}: 2 \lvert k\cdot\xi \lvert <(1-\eps)|k||\xi|\}.
\end{equation}
We now collect some results from \cite{Dam.2018}. A separate proof would be a lengthy repetition of that paper. Hence, we only briefly explain how to extend the results in \cite{Dam.2018}. 
\begin{lem}\label{lem:resultsthomas}
	Let $\Lambda\in(0,\infty]$ and $\xi\in\IR^\nu$. Then we have:
	\begin{enumerate}[(1)]
		\item\label{part:sigmaminimum} $\Sigma_{\Lambda}(\xi)\geq \Sigma_{\Lambda}(0)$.
		\item\label{lem:lowerboundsigma}$\Sigma_{\Lambda}(\xi-k) + \omega(k) \geq  \Sigma_\Lambda(\xi)$ for $k\in\IR^\nu$.
		\item\label{lem:lowerboundsigma2}$\Sigma_{\Lambda}(\xi-k) + \omega(k) >  \Sigma_\Lambda(\xi)$ for $k\notin\IR\xi$.
		\item\label{lem:lowerboundsigma3} For $\eps\in(0,1)$ there exist $D<1$ and $r>0$ independent of $\Lambda$ such that for all $k\in B_r(0)\cap S_\eps(\xi)$ we have
		$$ \Sigma_{\Lambda}	(\xi-k) -\Sigma_\Lambda(\xi) \ge -D\omega(k). $$
	\end{enumerate}
\end{lem}
\begin{proof}
	All statements are proved for finite $\Lambda$ in \cite{Dam.2018}. Using that $\Sigma_{\Lambda}(\xi)+E_\Lambda$ converges to $\Sigma_{\infty}(\xi)$ for all $\xi\in\IR^\nu$ by \cref{LargeFiberProp,lem:normresolventConv} the statements \cref{part:sigmaminimum,lem:lowerboundsigma} follow immediately. The proofs of \cref{lem:lowerboundsigma2,lem:lowerboundsigma3} are word by word the same as in \cite[Lemma 4.3]{Dam.2018}, as the proofs only rely on \cref{part:sigmaminimum,lem:lowerboundsigma}, \partref{hyp3}{part:increasing} and the rotation invariance of $\omega$ and $\Sigma_\Lambda$, where the latter follows from \cref{lem:rotation-operator}.
\end{proof}
\noindent
For $\xi\in\IR^\nu$ and $k\in \RR^\nu\setminus\{0\}$ we define 
\begin{equation}\label{def:UVECT}
\begin{aligned}
&P_{0}(\xi) = 1_{\{\Sigma_\infty(\xi)\}}(H_\infty(\xi)),\\
&V_i(\xi) = 2C_\omega (B_\xi^{1/2}P_0(\xi))^*(\xi_i-d\Gamma(m_i))B_\xi^{-1} (B_\xi^{1/2}P_0(\xi)) \quad\mbox{for}\ i=1,\ldots,\nu,\\
&k\cdot V(\xi) = k_1V_1(\xi)+...+k_\nu V_\nu(\xi),
\end{aligned}
\end{equation}
where $C_\omega=\lim\limits_{k\to0}\frac{|k|}{\omega(k)}$ as in \cref{hyp3}. Note, $B_\xi^{1/2}P_0(\xi)$ is bounded by \cref{LargeFiberProp,lem:propertiesOfB,lem:AohBxi}. Furthermore, the operator $(\xi_i-d\Gamma(m_i))B_\xi^{-1}$ is bounded by \cref{lem:Dxik} and selfadjoint (acts as a real multiplication operator in each $\cF^{(n)}$), so $V_i(\xi)$ is bounded and selfadjoint.

\begin{lem}\label{lem:U-invertible}
	Let $\xi\in\IR^\nu$, $\eps\in (0,1)$, $k\in S_\eps(\xi)$ and $\hat k=k/|k|$. Then we have $\lVert \hat k\cdot V(\xi) \lVert\leq \frac 12 (1-\eps)$ and hence, the operator $1-\hat k\cdot V(\xi)$ is invertible.
\end{lem}
\begin{proof}
	As $\hat k\cdot V(\xi)$ is selfadjoint, we have $$\lVert \hat k\cdot V(\xi) \lVert = \sup_{\substack{\psi\in \cF\\\lVert \psi \lVert=1}  }\lvert \braket{\psi,\hat k\cdot V(\xi) \psi}\lvert.$$ Let $\psi\in \cF$ and assume  $\lVert \psi \lVert=1$. We note that $$\braket{\psi,\hat k\cdot V(\xi) \psi}=\braket{P_0(\xi)\psi,\hat k\cdot V(\xi) P_0(\xi)\psi}.$$ Since the statement follows trivially if $P_0(\xi)\FS=\{0\}$, we may assume $P_0(\xi)\psi = \psi$ w.l.o.g and define $v_\psi(\xi) = 2(\xi-q_{d\Gamma(m)}(\psi,\psi))$. This implies
	\begin{align*}
	\braket{\psi,\hat k\cdot V(\xi) \psi} = 2C_\omega \hat k\cdot ( \lVert \psi\lVert^2 \xi - q_{d\Gamma(m)}(\psi,\psi))= C_\omega\hat k\cdot v_\psi(\xi).
	\end{align*}
	Hence, it suffices to prove
	\begin{equation}\label{eq:vpsibound}
	|\hat k\cdot v_\psi(\xi)| \le \frac{1-\eps}{2C_\omega}.
	\end{equation}
	By \partref{thm:operatordomain}{part:quadraticdomain} and \cref{lem:qDgammaM}, we see
	\begin{align*}
	q_{H_\infty(\xi+h)}(\psi,\psi) &= q_{H_\infty(\xi)}(\phi,\phi) + 2h\cdot (\xi-q_{d\Gamma(m)}(\psi,\psi))+|h|^2\|\psi\|^2.\\&=\Sigma_\infty(\xi) + h\cdot v_\psi(\xi) + |h|^2
	\end{align*}
	Using $q_{H_\infty(\xi+h)}(\psi,\psi)\ge \Sigma_{\infty}(\xi+h)$ yields
	\begin{equation}
	\label{eq:transformationupperbound}
	\Sigma_\infty(\xi +h) - \Sigma_\infty(\xi) \le h\cdot v_\psi(\xi) + |h|^2 \qquad\mbox{for all }h\in\IR^\nu.
	\end{equation} 
	If $\xi = 0$, the left hand side is non-negative by \partref{lem:resultsthomas}{part:sigmaminimum}, so taking the limit $|h|\to 0$ for fixed $\hat h$ we obtain $\hat h\cdot v_\psi(0)\ge 0$ for all $h\in\IR^\nu$. This directly implies $v_\psi(0)=0$ and hence \cref{eq:vpsibound}. 
	
	From now we can assume $\xi\ne 0$. By \partref{lem:resultsthomas}{lem:lowerboundsigma}, we know that \[\Sigma_\infty(\xi+h)-\Sigma_\infty(\xi)\ge -\omega (h).\] Inserted into \cref{eq:transformationupperbound} this leads to
	$$ h\cdot v_\psi(\xi) \ge - \omega(h) -|h|^2\qquad\mbox{for all}\ h\in \RR^\nu.$$
	For $h\in\IR^\nu\setminus\{0\}$, we divide by $|h|$ and again take $|h|\to 0$ at fixed $\hat h$ to obtain
	$ \hat h\cdot v_\psi(\xi) \ge -C_\omega^{-1}$  and hence
	$ |v_\psi(\xi)|\le C_\omega^{-1}$. Let $O\in\IR^{\nu\times \nu}$ be orthogonal and define $\phi=\Gamma(U_O)\psi$. Then $\phi\in \cQ(H_\infty(\xi))$ by \cref{lem:rotation-operator}. Using \partref{thm:operatordomain}{part:quadraticdomain} together with \cref{lem:rotation-operator,lem:qDgammaM} we get the two expressions 
	\begin{align*}
	\Sigma_\infty(\xi)& \leq q_{H_\infty(\xi)}(\phi,\phi) = q_{H_\infty(O\xi)}(\psi,\psi) = q_{H_\infty(0)}(\psi,\psi) - 2 O\xi \cdot q_{d\Gamma(m)}(\psi,\psi)+|\xi|^2\\
	\Sigma_\infty(\xi)& = q_{H_\infty(\xi)}(\psi,\psi) = q_{H_\infty(0)}(\psi,\psi) - 2\xi \cdot q_{d\Gamma(m)}(\psi,\psi)+|\xi|^2.
	\end{align*}
	from which we obtain
	%
	$$ \xi\cdot q_{d\Gamma(m)}(\psi, \psi) \ge O\xi\cdot q_{d\Gamma(m)}(\psi,\psi) \qquad\mbox{for all orthogonal }O\in\IR^{\nu\times\nu}. $$
	Hence, there is a constant $R_\psi\in\IR$ such that $q_{d\Gamma(m)}(\psi, \psi)=R_\psi\xi$. For all $k\in S_\eps(\xi)$ this means
	\[ |\hat k\cdot v_\psi(\xi)| = 2|(1-R_\psi)\hat k\cdot \xi| \le (1-\eps)|\hat k||(1-R_\psi)\xi| = \frac{1-\eps}{2}|v_\psi(\xi)| \le \frac{1-\eps}{2C_\omega}.\qedhere \]
\end{proof}
\noindent
We will also need, that $(1-\hat k \cdot V(\xi))^{-1}$ conserves weak convergence to $0$.
\begin{lem}\label{lem:weakConv}
	Let $\xi\in\IR^\nu$, $R>0$ and $\eps\in (0,1)$.
	Further, let $\{o(k)\}_{k\in B_R(0)\cap S_{\eps}(\xi)}$ be a bounded family of operators satisfying $\wlim\limits_{k\to 0} o(k) = 0$. Then 
	$$ \wlim\limits_{k\to 0} (1-\hat k\cdot V(\xi))^{-1}o(k) = 0. $$
\end{lem}
\begin{proof}
	Let $\phi,\psi\in \cF$. By \cref{lem:U-invertible}, we know $\lVert \hat k\cdot V(\xi) \lVert<\frac 12(1-\varepsilon)$  so
	\begin{equation*}
	\langle \psi, (1-\hat k\cdot V(\xi))^{-1}o(k)\phi\rangle = \sum_{n=0}^{\infty} \langle (\hat k \cdot V(\xi))^n\psi, o(k)\phi\rangle
	\end{equation*}
	By uniform boundedness and dominated convergence, it is enough to see each term in the sum converges to 0 as $k\to 0$. This follows from $|\hat k|=1$ and
	\begin{equation*}
	\langle (\hat k \cdot V(\xi))^n\psi, o(k)\phi\rangle =\sum_{i_1=1}^{\nu}...\sum_{i_n=1}^{\nu} \hat k_{i_1}...\hat k_{i_n}\langle V_{i_1}(\xi)...V_{i_n}(\xi)\psi, o(k)\phi\rangle.\qedhere
	\end{equation*}
\end{proof}
\noindent
For $\xi\in\IR^\nu$ and $\Lambda \in [0,\infty]$ we introduce the operators
\begin{equation}
\begin{aligned}
&Q_{0,\Lambda}(k,\xi) = \omega(k) (H_{\Lambda}(\xi)-\Sigma_\Lambda(\xi)+\omega(k))^{-1} \qquad\mbox{for}\ k\in\IR^\nu\setminus\{0\},\\
&Q_{\Lambda}(k,\xi) = \omega(k) (H_{\Lambda}(\xi-k)-\Sigma_\Lambda(\xi)+\omega(k))^{-1} \qquad\mbox{for}\ k\in\IR^\nu\setminus\IR\xi,
\end{aligned}	
\end{equation}
which are well defined and bounded by \partref{lem:resultsthomas}{lem:lowerboundsigma2}. We will write $Q_0$ and $Q$ instead of $Q_{0,\infty}$ and $Q_\infty$, respectively. The next lemmas collect some simple statements about these operators.
\begin{lem}\label{lem:Q0boundandlim}
	Let $\xi\in\IR^\nu$ and $R>0$. Then the operator $B_\xi^{1/2} Q_{0}(k,\xi)$ is bounded uniformly in $k\in B_R(0)\setminus\{0\}$ and
	$$ \slim_{k\to 0} B_\xi^{1/2} Q_{0}(k,\xi)(1-P_{0}(\xi)) = 0. $$
\end{lem}
\begin{proof}
	This follows from the domain inclusion \cref{eq:domains} and \cref{Thm:ProjectionConvergence}.
\end{proof}
\begin{lem}\label{lem:Quniformbound}
	Let $\eps\in(0,1)$ and $r>0$ as in \partref{lem:resultsthomas}{lem:lowerboundsigma3}. Then the operators $Q(k,\xi)$ and $B_\xi^{1/2} Q(k,\xi)$ are bounded uniformly for all $k\in B_r(0)\cap S_\eps(\xi)$.
\end{lem}
\begin{proof}
	For all $k\in B_r(0)\cap S_\eps(\xi)$, \partref{lem:resultsthomas}{lem:lowerboundsigma3} yields $\lVert Q(k,\xi)\lVert\leq (1-D)^{-1}$ for some $D\in(0,1)$, which proves the first uniform upper bound.
	
	Now note $B_\xi^{1/2} Q(k,\xi)$ is bounded by \cref{eq:domains}. By \cref{LargeFiberProp} and \cref{lem:AohBxi}, we can pick $\lambda$ small enough such that $k\mapsto \lVert B_\xi^{1/2}(H_\infty(\xi-k)-\lambda)^{-1}\lVert$ is continuous and hence uniformly bounded by some constant $C$ on $B_r(0)$. This leads to
	\begin{align*}
	\|B_\xi^{1/2} Q(k,\xi)\| \le C \|(H_\infty(\xi-k)-\lambda)Q(k,\xi)\| .
	\end{align*}
	The uniform bound on $B_\xi^{1/2}Q(k,\xi)$ now follows from the one on $Q(k,\xi)$ by
	\[ (H_\infty(\xi-k)-\lambda)Q(k,\xi) = \omega(k) + (\Sigma_\infty(\xi)-\omega(k)-\lambda)Q(k,\xi).\qedhere \]
\end{proof}

\begin{lem}
	\label{lem:Q-formula}
	For $\xi\in\IR^\nu$, $k\in\IR^\nu\setminus \IR\xi$ and $\Lambda\in(0,\infty]$ we have
	\begin{equation*}
	Q_\Lambda(k,\xi) = Q_{0,\Lambda}(k,\xi) + \frac{1}{\omega(k)}(B_\xi^{1/2}Q_{0,\Lambda}(k,\xi))^*D_{\xi}(k)B_{\xi}^{-1} (B_{\xi}^{1/2}Q_\Lambda(k,\xi)).
	\end{equation*} 
\end{lem}
\begin{proof}
	The statement for $\Lambda<\infty$ follows from the resolvent identity, \cref{eq:Dxikdiff} and the fact that $B_\xi$ and $D_{\xi}(k)$ commute strongly. It remains to show we can take weak limits on both sides.
	To that end, it suffices to prove
	\begin{equation}\label{eq:Q-stronglimits}
	\slim_{\Lambda\to\infty}B_\xi^{a/2}Q_{0,\Lambda}(k,\xi)=B_\xi^{a/2}Q_0(k,\xi) \  \mbox{and}\  \slim_{\Lambda\to\infty}B_\xi^{a/2}Q_{\Lambda}(k,\xi)=B_\xi^{a/2}Q(k,\xi)
	\end{equation}
	for $a\in\{0,1\}$. By \cref{lem:normresolventConv,LargeFiberProp}, we have
	$ \lim\limits_{\Lambda\to\infty}\Sigma_\Lambda(\xi)+E_\Lambda=\Sigma_\infty(\xi)$,
	so using
	\begin{align*}
	\exp(-t(H_\Lambda(\xi+h)-\Sigma_\Lambda(\xi)))&=\exp(-t(\Sigma_\Lambda(\xi)+E_\Lambda))\exp(-t(H_\Lambda(\xi+h)+E_\Lambda))\\& \xrightarrow{\Lambda \rightarrow \infty} \exp(-t(H_\infty(\xi+h)-\Sigma_\infty(\xi)))
	\end{align*}
	so the case $a=0$ in \cref{eq:Q-stronglimits} follows due to \cref{lem:normresolventConv}.
	
	Further,
	for $(Z_\Lambda,h)\in\{(Q_{0,\Lambda}(k,\xi),\xi),(Q_{\Lambda}(k,\xi),\xi-k)\}$ and $\lambda\in\IR$ chosen such that  $\lambda<\Sigma_\Lambda(\xi)+E_\Lambda$ for all $\Lambda>0$ and $\xi\in\IR^\nu$, the resolvent identity yields
	\begin{align*}
	B_\xi^{1/2}Z_\Lambda=&
	\ \omega(k)B_\xi^{1/2}(H_\Lambda(h)+E_\Lambda-\lambda)^{-1} \\
	& +  \omega(k) B_\xi^{1/2}(H_{\Lambda}(h)+E_\Lambda-\lambda)^{-1}\left(1-\frac{1}{\omega(k)}(\Sigma_\Lambda(\xi)+E_\Lambda-\lambda)\right)Z_\Lambda.
	\end{align*}
	Hence, \cref{eq:Q-stronglimits} follows for the case $a=1$ by \cref{LargeFiberProp,lem:AohBxi}.
\end{proof}
\begin{lem}
	\label{lem:P-conv}
	Let $\xi\in\IR^\nu$ and $\eps\in(0,1)$. Then
	$$ \wlim_{\substack{k\to0\\k\in S_\eps(\xi)}} Q(k,\xi)(1-P_0(\xi)) = \wlim_{\substack{k\to0\\k\in S_\eps(\xi)}} (1-P_0(\xi))Q(k,\xi) = 0. $$
\end{lem}
\begin{proof}
	By taking adjoints it suffices to prove one of the statements. By taking the adjoint in \cref{lem:Q-formula}, we notice
	\begin{align*}
	&Q(k,\xi)(1-P_0(\xi)) = (Q(k,\xi))^*(1-P_0(\xi)) \\&\quad= Q_0(k,\xi)(1-P_0(\xi)) + \frac{|k|}{\omega(k)}(B_\xi^{1/2} Q(k,\xi))^*\frac{D_\xi(k)B_\xi^{-1}}{|k|} B_\xi^{1/2} Q_0(k,\xi)(1-P_0(\xi)).
	\end{align*}
	This goes to 0 strongly taking the limit $k\to0$ for $k\in S_\eps(\xi)$ by \cref{lem:Dxik,lem:Q0boundandlim,lem:Quniformbound} and \partref{hyp3}{part:lim}.
\end{proof}
\noindent
We now prove the main ingredient of our proof, explicitly, that $(1-\hat k\cdot V(\xi))^{-1}P_0(\xi)$ describes the behaviour of $Q(k,\xi)$ for small $k$ in the weak sense.
\begin{lem}
	\label{lem:final}
	Let $\xi\in\IR^\nu$ and $\eps\in(0,1)$. Then
	$$\wlim_{\substack{k\to0\\k\in{S}_\eps(\xi)}}(Q(k,\xi)-(1-\hat{k}\cdot V(\xi))^{-1} P_0(\xi)) = 0.$$
\end{lem}
\begin{proof} Throughout this proof we assume $k\in S_\eps(\xi)$. First, note
	$$ \wlim_{k\to 0}(Q(k,\xi) - P_0(\xi) Q(k,\xi)P_0(\xi))=0, $$
	by \cref{lem:P-conv}. Using $Q_{0}(k,\xi)P_0(\xi)=P_0(\xi)$, $P_0(\xi)(B_\xi^{1/2} Q_{0}(k,\xi))^*=(B_\xi^{1/2} P_0(\xi))^*$ and \cref{lem:Q-formula}, we get
	\begin{align*}
	P_0(\xi)Q(k,\xi)P_0(\xi) & = P_0(\xi) + \frac{1}{\omega(k)}(B_\xi^{\frac{1}{2}} P_0(\xi))^* D_\xi(k)B_\xi^{-\frac{1}{2}}Q(k,\xi)P_0(\xi)\\
	&= P_0(\xi) + \frac{1}{\omega(k)}(B_\xi^{\frac{1}{2}}P_0(\xi))^*D_\xi(k)B_\xi^{-\frac{1}{2}}P_0(\xi)Q(k,\xi)P_0(\xi) + o_1(k),\\
	& = P_0(\xi) + \hat k\cdot V(\xi)P_0(\xi) Q(k,\xi)P_0(\xi) + o_1(k)+o_2(k), \quad\mbox{where}\\
	o_1(k) :=& \frac{1}{\omega(k)} (B_\xi^{\frac{1}{2}}P_0(\xi))^* D_\xi(k)B_\xi^{-\frac{1}{2}}(1-P_0(\xi))Q(k,\xi)P_0(\xi) \quad\mbox{and}\\
	o_2(k) := &\left(\frac{1}{\omega(k)}(B_\xi^{\frac{1}{2}}P_0(\xi))^*D_\xi(k)B_\xi^{-\frac{1}{2}}P_0(\xi)-\hat k\cdot V(\xi)\right)P_0(\xi)Q(k,\xi)P_0(\xi).
	\end{align*}
	This leads to
	$$ P_0(\xi)Q(k,\xi)P_0(\xi) - (1-\hat{k}\cdot V(\xi))^{-1} P_0(\xi) =(1-\hat k\cdot V(\xi))^{-1} (o_1(k)+o_2(k)). $$
	By \cref{lem:weakConv}, it suffices to prove $\wlim\limits_{k\to 0} o_1(k) = \wlim\limits_{k\to 0} o_2(k) = 0$. Let $\phi,\psi\in\FS$. By the definition \cref{defn:Dxik}, we have
	$$ (B_\xi^{\frac{1}{2}}P_0(\xi))^*D_\xi(k)B_\xi^{-\frac{1}{2}}(1-P_0(\xi))\psi =-2\sum_{i=1}^{\nu}k_i (B_\xi^{\frac{1}{2}}P_0(\xi))^*d\Gamma(m_i)B_\xi^{-\frac{1}{2}}(1-P_0(\xi))\psi $$
	and hence
	$$\braket{\phi,o_1(k)\psi} = -2 \frac{|k|}{\omega(k)}\sum_{i=1}^{\nu}\hat k_i \langle B_\xi^{\frac 12} P_0(\xi)\phi, d\Gamma(m_i)B_\xi^{-1}B_\xi^{\frac 12}(1-P_0(\xi))Q(k,\xi)P_0(\xi)\psi\rangle. $$
	Note $B_\xi^{\frac12}(1-P_0(\xi))Q(k,\xi)$ is uniformly bounded by \cref{lem:Quniformbound} and the fact $B_\xi^{\frac 12}P_0(\xi)$ is bounded by \cref{eq:domains}. Hence can apply \cref{lem:P-conv,WeakStrongConvergence} \cref{part:weakconvergenceproduct}, to see $\braket{\phi,o_1(k)\psi}\xrightarrow{k\to 0} 0$. 
	Further, definition \cref{def:UVECT} yields
	$$\hat k\cdot V(\xi) = C_\omega|k|^{-1}(B_\xi^{\frac{1}{2}}P_0(\xi))^*D_\xi(k)B_\xi^{-\frac{1}{2}}P_0(\xi)-C_\omega|k|P_0(\xi), $$
	so we have
	\begin{align*}
	o_2(k) =  \left(\frac{|k|}{\omega(k)}-C_\omega\right)&(B_\xi^{\frac{1}{2}}P_0(\xi))^*\frac{D_\xi(k)B_\xi^{-1}}{|k|}B_\xi^{\frac{1}{2}}P_0(\xi)Q(k,\xi)P_0(\xi) \\&  - |k|C_\omega P_0(\xi)Q(k,\xi)P_0(\xi).
	\end{align*}
	This converges to 0 in norm, due to boundedness of $B_\xi^{\frac 12}P_0(\xi)$, \cref{lem:Quniformbound}, \cref{lem:Dxik} and \partref{hyp3}{part:lim}.
\end{proof}
\noindent
The rest of the proof is an adaption of the proof in \cite{Dam.2018}.
\begin{proof}[\textbf{Proof of \cref{thm:nogs}}.]\ \\
	The proof goes by contradiction. We fix $\xi\in\IR^\nu$ and assume there exists a $\psi_{gs}\in\cD(H_\infty(\xi))$ such that $\|\psi_{gs}\|=1$ and $H_\infty(\xi)\psi_{gs}=\Sigma_\infty(\xi)\psi_{gs}$.
	
	For $k\in\IR^\nu$ we define the pointwise annihilation operator $a_k$ which acts on $\psi^{(n)}\in\FS^{(n)}$ by
	\begin{equation}\label{eq:annihilator}
	a_k\psi^{(n)} = \sqrt{n}\psi^{(n)}(k,\cdots)\in\FS^{(n-1)}.
	\end{equation}
	By the Fubini-Tonelli theorem, this expression is well-defined for almost every $k\in\IR^\nu$.
	In \cref{Lem:pullthr}, we prove $a_k\psi_{gs} = (a_k\psi^{(n+1)}_{gs})_{n=0}^\infty\in\FS$ for almost all $k\in\IR^\nu$ and the pull-through formula
	\begin{equation}\label{eq:pullthr}
	a_k\psi_{gs} = -\frac{v(k)}{\omega(k)}Q(k,\xi)\psi_{gs} \quad\mbox{for almost all}\ k\in\IR^\nu.
	\end{equation}
	Pick $\eps = \frac{1}{2}$ and let $k\in S_\eps(\xi)$. Then $\lVert \hat{k}\cdot V(\xi) \lVert\le\frac{1}{4}$ by \cref{lem:U-invertible}, so a power expansion shows
	\begin{equation}\label{eq:smart}
	\lVert 1-(1-\hat{k}\cdot V(\xi))^{-1}\lVert\leq \frac{\lVert \hat{k}\cdot V(\xi) \lVert} {1-\lVert \hat{k}\cdot V(\xi) \lVert}=-1+\frac{1} {1-\lVert \hat{k}\cdot V(\xi) \lVert}\leq \frac{1}{3}
	\end{equation}
	We denote the number operator $N=d\Gamma(1)$ and choose a normalized element $\eta\in\cD(N^{1/2})$ such that $\lvert \langle \eta,\psi_{gs} \rangle\lvert>\frac{1}{2}$. Then \cref{eq:pullthr} shows
	$$ \braket{\eta,a_k\psi_{gs}} = -\frac{v(k)}{\omega(k)}\braket{\eta,Q(k,\xi)\psi_{gs}} \quad\mbox{for almost every}\ k\in\IR^\nu. $$
	Further, for $k\in{S}_\eps(\xi)$ \cref{lem:final} yields
	$$\braket{\eta,Q(k,\xi)\psi_{gs}}-\braket{\eta,(1-\hat{k}\cdot V(\xi))^{-1}\psi_{gs}} \xrightarrow{k\to 0} 0.$$
	By  $\cref{eq:smart}$ we now see $\lvert \braket{\eta,(1-\hat{k}\cdot V(\xi))^{-1}\psi_{gs}}\lvert > \frac{1}{2}-\frac{1}{3}=\frac{1}{6}$. Hence, given any $C\in(0,\frac 16)$, there is $R>0$ such that $$|\braket{\eta,a_k\psi_{gs}}|\ge C \frac{|v(k)|}{\omega(k)} \qquad\mbox{for almost all}\ k\in {S}_\eps(\xi)\cap B_R(0) =: \tilde B_{\eps,R}(\xi).$$
	Further, since ${S}_\eps(\xi)$ is open, non-empty (due to $\nu\ge2$) and invariant under positive scalings, we obtain by rotation invariance of $v$ and $\omega$ that
	$$ \int_{\tilde B_{\eps,R}(\xi)}\frac{|v(k)|^2}{\omega(k)^2} dk = \frac{\operatorname{vol}(\tilde B_{\eps,1}(\xi) )}{\operatorname{vol}(B_1(0))}\int_{B_R(0)}\frac{|v(k)|^2}{\omega(k)^2} dk = \infty. $$
	This proves that $\braket{\eta,a_k\psi_{gs}}$ is not square-integrable.
	
	On the other hand using \cref{Prop:Bregninger} and Cauchy-Schwarz we find
	\begin{align*}
	|\braket{\eta,a_k\psi_{gs}}|^2&\le \|(N+1)^{1/2}\eta\|^2\|(N+1)^{-1/2}a_k\psi_{gs}\|^2\\
	\quad&=\|(N+1)^{1/2}\eta\|^2\sum_{n=1}^\infty \int_{\IR^{(n-1)\nu}}|\psi^{(n)}_{gs}(k,k_1,\ldots,k_{n-1})|^2dk_1\cdots dk_n,
	\end{align*}
	which is integrable by definition of the Fock space norm. Hence, we have arrived at a contradiction and the state $\psi_{gs}$ cannot exist.
\end{proof}

\subsection*{Acknowledgments}

The authors thank Oliver Matte and an anonymous referee for a variety of helpful comments and suggestions. Further, B.H. thanks Aarhus Universitet and especially Jacob Schach M\o{}ller for their kind hospitality, which enabled a large part of the work on this article.

\appendix
\section{Some Lemmas from General Operator Theory}
\label{app:operator}
In this appendix we state well-known or easy to prove properties of selfadjoint operators.
Without further mentioning, we assume $\fH$ is a separable Hilbert space.

\begin{lem}\label{lem:domainsqrt}
	Let $A,B$ be selfadjoint operators in $\fH$ and $U$ be unitary. If $U\cD(A)\subset\cD(B)$ then $U\cQ(A)\subset \cQ(B)$.
\end{lem}
\begin{proof}
	Note $\cD(U\lvert A\lvert U^*)=\cD(UAU^*)=U\cD(A)\subset \cD(B)=\cD(\lvert B \lvert)$, which by \cite[Theorem 9.4]{Weidmann.1980} implies $\cD(U\lvert A\lvert^{1/2} U^*)\subset \cD(\lvert B \lvert^{1/2})=\cQ(B)$. As $\cD(U\lvert A\lvert^{1/2} U^*)=U\cQ(A)$ the claim follows
\end{proof}

\begin{lem}\label{commuteingCore}
	Let $A$ and $B$ be strongly commuting and selfadjoint operators. If $A$ is $B$-bounded and $\cD$ is a core for $B$ then $\cD$ is a core for $A$.
\end{lem}
\begin{proof}
	Clearly any element in $\cD(B)$ can be approximated in $A$-norm by elements in $\cD$, so it is enough to see $\cD(B)$ is a core for $A$. However, for any $\psi\in \cD(A)$ we can choose the approximating sequence $1_{\{\lvert B\lvert \leq n\}}\psi$, which converges in $A$-norm since $A1_{\{\lvert B\lvert \leq n\}}\psi=1_{\{ \lvert B\lvert \leq n\}}A\psi$.
\end{proof}

\begin{lem}
	\label{Thm:ProjectionConvergence}
	Let $A,B$ be selfadjoint operators on $\fH$ and $\omega:\IR^\nu\to \IR$. Assume $A$ is bounded below, $B$ is $A$-bounded, $\omega$ is continuous, $\omega(0)=0$ and $\omega(k)>0$ for $k\ne0$. Define $\lambda=\inf(\sigma(A))$ and $f(k)=\omega(k)B(A-\lambda+\omega(k))^{-1}$ for $k\ne 0$. Then map $f:\IR^\nu\setminus\{0\}\to \cB(\fH)$ is bounded on $B_R(0)\setminus\{0\}$ for any $R>0$ and
	$\slim\limits_{k\to 0}f(k)= BP_A(\{\lambda\}),$ where $P_A$ is the projection valued measure associated to $A$.
\end{lem}
\begin{proof}
	By the assumptions $C=B(A-\lambda+1)^{-1}$ is bounded. The conclusion now follows from the spectral theorem and the identity
	\[f(k)=\omega(k)C+C\omega(k)(A-\lambda+\omega(k))^{-1}-\omega(k)C\omega(k)(A-\lambda+\omega(k))^{-1}.\qedhere\]
\end{proof}

\begin{lem}[{\cite[Theorem 5.37]{Weidmann.1980}}]
	\label{lem:operatorandformdomain}
	Let $A$ be a selfadjoint operator on $\fH$ and $\psi\in\cQ(A)$. Then the following are equivalent:
	\begin{enumerate}[(1)]
		\item $\psi\in\cD(A)$,
		\item The map $\cQ(A) \ni \phi\mapsto q_A(\psi,\phi)$ is continuous w.r.t. the subspace topology of $\fH$.
		\item There is a form core $D$ of $A$ such that $D\ni \phi\mapsto q_A(\psi,\phi)$ is continuous w.r.t. the subspace topology of $\fH$.
	\end{enumerate}
\end{lem}
\begin{lem}[{\cite[Theorem 6.25]{Teschl.2014}}]\label{lem:resolventformbound}
	Let $A$ be a selfadjoint operator with $A\geq \lambda$, $q$ a symmetric sesquilinear form with $\cQ(A)\subset\cQ(q)$ and $a,b\in\IR$.
	The symmetric sesquilinear form $q((A-z)^{-1/2}\phi,(A-z)^{-1/2}\psi)$ for $\phi,\psi\in\fH$ corresponds to a bounded operator $C(z)$ with $\|C(z)\|\le a$ for $z<-ba^{-1}-\lambda$ if and only if
	$$ q(\psi,\psi) \le aq_A(\psi,\psi) + b\|\psi\|^2 \qquad\mbox{for all}\ \psi\in\cQ(q).$$
	Further, if $a<1$, then
	$$ (B-z)^{-1} = (A-z)^{-1/2}(1+C(z))^{-1}(A-z)^{-1/2}, $$
	where $B$ denotes the selfadjoint and lower bounded operator corresponding to $q_A+q$.
\end{lem}
\begin{lem}\label{WeakStrongConvergence}
	Let $(C_n)\subset \cB(\fH)$, $C\in \cB(\fH)$ and $A$ be a densely defined closed operator on $\fH$ such that $\cD(A)\supset C_n\fH$ for all $n\in\IN$ and $(AC_n)$ is uniformly bounded. Then: 
	\begin{enumerate}[(1)]
		\item\label{part:weakconvergenceproduct}
		If $\wlim\limits_{n\to\infty}C_n = C$, then $C\fH\subset \cD(A)$ and $\wlim\limits_{n\to\infty} AC_n = AC$.
		\item\label{part:strongconvergenceproduct} 
		If $\slim\limits_{n\to\infty}C_n^* = C^*$, then $C\fH\subset \cD(A)$ and $\slim\limits_{n\to\infty} (AC_n)^* = (AC)^*$.
	\end{enumerate}
	\begin{proof}
		First, assume $\wlim\limits_{n\to\infty} C_n =C$ and 
		pick $B>0$ such that $\lVert AC_n\lVert < B$ holds for all $n\in \NN$. For $\psi\in \fH, \phi\in \cD(A^*)$, we find
		\begin{equation*}
		\lvert \braket{A^*\phi,C\psi} \lvert=\lim_{n\rightarrow  \infty} \lvert \braket{A^*\phi,C_n\psi} \lvert\leq B\lVert \phi\lVert \lVert \psi\lVert,
		\end{equation*}
		so $C\psi\in \cD(A)$. Clearly
		\begin{equation*}
		\braket{\phi,AC\psi}=\braket{A^*\phi,C\psi} =\lim_{n\rightarrow  \infty}  \braket{\phi,AC_n\psi}.
		\end{equation*}
		Since $\cD(A)$ is dense and the sequence $(AC_n\psi)$ is bounded by $B\lVert \psi\lVert$, this implies weak convergence of $(AC_n\psi)$ to $AC\psi$, cf. \cite[Theorem 4.24]{Weidmann.1980}, and hence proves \cref{part:weakconvergenceproduct}.
		
		To prove \cref{part:strongconvergenceproduct}, we first note that $\slim\limits_{n\to\infty} C_n^* = C^*$ implies $\wlim\limits_{n\to\infty}C_n = C$, so \cref{part:weakconvergenceproduct} proves $C\fH\subset \cD(A)$.
		Further, we have $\lVert (AC_n)^*\lVert=\lVert AC_n\lVert$, so $((AC_n)^*)$ is uniformly bounded. For $\psi\in \cD(A^*)$, we have 
		\begin{equation*}
		(AC)^*\psi=C^*A^*\psi=\lim_{n\rightarrow  \infty}  C_n^*A^*\psi,
		\end{equation*}
		so the conclusion follow from \cite[Theorem 4.23]{Weidmann.1980}.
	\end{proof}
\end{lem}
\begin{lem}[{\cite[Theorems VIII.20, VIII.23 and VIII.24]{ReedSimon.1972} and \cite[Lemma 5.5]{DamMoller.2018b}}]\label{lem:normresolventConv}\ \\
	Let $(A_n)$ be a sequence of selfadjoint operators and $A$ a selfadjoint operator.
	\begin{enumerate}[(1)]
		\item\label{convSpec} Assume $(A_n)$ converges to $A$ in the norm resolvent sense. If $\lambda\notin \sigma(A)$ then $\lambda\notin \sigma(A_n)$ for $n$ large enough and $(A_n-\lambda)^{-1}$ converges to $(A-\lambda)^{-1}$ in norm.
		\item\label{convCalc} If $(A_n)$ converges to $A$ in the norm resolvent sense and $f:\RR\rightarrow \RR$ is continuous and bounded then $f(A_n)$ converges strongly to $f(A)$. If $f$ is vanishing at $\pm \infty$ then convergence is in norm.
		\item\label{convBound} If $(A_n)$ uniformly bounded below by $\lambda\in\IR$ then $(A_n)$ converges to $A$ in the norm resolvent sense if and only if $e^{-tA_n}$ converges to $e^{-tA}$ in norm for all $t>0$. In this case, $\inf\sigma(A_n)$ converges to $ \inf\sigma(A)$.
	\end{enumerate}
\end{lem}
\begin{lem}\label{Thm:ConvRelBound}
	Let $(A_n)$ be a sequence of selfadjoint operators on $\fH$ and assume there is $\lambda\in\IR$ such that $A_n\geq \lambda$ for all $n\in \NN$. Let $A$ and $B$ be selfadjoint operators on $\fH$ and assume that $\cQ({A_n})\subset \cQ(B)$ for all $n\in \NN$, that $(A_n)$ converges to $A$ in the norm resolvent sense and $\lvert B\lvert^{1/2}$ has a bounded inverse. Define for $z< \inf\sigma(A_n)$ the bounded operator $C_{n,z}=\lvert B\lvert^{1/2}(A_n-z)^{-1/2}$.
	If $C_{n,z}C_{n,z}^*$ converges strongly for some $z<\lambda$, then $\cQ(A)\subset \cQ(B)$ and $\slim\limits_{n\to\infty}{C_{n,z}}=|B|^{1/2}(A-z)^{-1/2}=:C_{\infty,z}$ for all $z<\inf\sigma(A)$.
\end{lem}
\begin{proof}
	Note $A \geq \lambda$ by \partref{lem:normresolventConv}{convBound}. Pick $z_0<\lambda$ such that $C_{n,z_0}C_{n,z_0}^*$ converges to a selfadjoint operator $C\in\cB(\fH)$ strongly. For $\phi\in \cQ(B)$ and $\psi\in \fH$ we see
	\begin{equation*}
	\langle |B|^{1/2}\phi,(A-z_0)^{-1/2}\psi \rangle \leq \lim_{n\rightarrow \infty}\langle \phi,C_{n,z_0}C_{n,z_0}^*\phi \rangle^{1/2}\lVert \psi \lVert\leq \lVert C\lVert^{1/2}\lVert \psi\lVert\lVert \phi\lVert,
	\end{equation*}
	showing $(A-z_0)^{-1/2}\psi \in \cQ(B)$ and hence $\cQ(A)\subset \cQ(B)$.
	
	For $\phi,\psi\in \cQ(B)$ the norm resolvent convergence of $(A_n)$ also yields
	\begin{align*}
	\langle \phi,C_{\infty,z_0}C_{\infty,z_0}^*\psi \rangle &= \lim_{n\rightarrow \infty}\langle (A_n-z_0)^{-1/2}|B|^{1/2}\phi,(A_n-z_0)^{-1/2}|B|^{1/2}\psi \rangle\\&= \lim_{n\rightarrow \infty}\langle \phi,C_{n,z_0}C_{n,z_0}^*\psi \rangle=\langle \phi,C\psi \rangle
	\end{align*}
	so $C_{\infty,z_0}C_{\infty,z_0}^*=C$.
	
	Note $\lVert C_{n,z_0} \lVert^2= \lVert C_{n,z_0}C_{n,z_0}^* \lVert$ is bounded uniformly in $n$ by the uniform boundedness principle. Since $\cQ(A)$ is dense, it is now enough to show $\lim\limits_{n\to\infty}C_{n,z_0}\psi=C_{\infty,z_0}\psi$ for all $\psi\in\cQ(A)$ by \cite[Theorem 4.23]{Weidmann.1980}.  Hence, using
	\begin{align*}
	C_{n,z_0}\psi =& C_{n,z_0}C_{n,z_0}^*\lvert B\lvert ^{-1/2}(A-z_0)^{1/2}\psi\\&+C_{n,z_0}((A-z_0)^{-1/2}-(A_n-z_0)^{-1/2}) (A-z_0)^{1/2}\psi,
	\end{align*}
	we see that $C_{n,z_0}\psi$ converges to $C\lvert B\lvert ^{-1/2}(A-z_0)^{1/2}\psi=C_{\infty,z_0}\psi$ for all $\psi\in\cQ(A)$ by \partref{lem:normresolventConv}{convCalc}. For any other $z<\inf\sigma(A)$ we conclude that $z<\inf\sigma(A_n)$ for $n$ large enough by \partref{lem:normresolventConv}{convBound}. Then, by \partref{lem:normresolventConv}{convCalc}
	\begin{equation*}
	C_{n,z}=C_{n,z_0}\left (\frac{A_n-z_0}{A_n-z}\right )^{1/2} \xrightarrow{s} C_{\infty,z_0}\left (\frac{A-z_0}{A-z}\right )^{1/2}=C_{\infty,z}.\qedhere
	\end{equation*}
\end{proof}
\noindent
For two selfadjoint operators $A$ and $B$ we define the sesquilinear form associated to their commutator as
\begin{equation}\label{defn:commutatorform}
q_{[A,B]}(\psi,\phi) = \braket{A\psi,B\phi}-\braket{B\psi,A\phi} \qquad\mbox{for}\ \psi,\phi\in\cD(A)\cap\cD(B).
\end{equation}
\begin{lem}\label{lem:commutator}
	Let $A$ and $B$ be selfadjoint operators and assume there is a set $D\subset \cD(B)$ such that $e^{itB}D\subset \cD(A)$ for all $t\in\IR$ and $t\mapsto Ae^{itB}\psi$ is continuous for all $\psi\in D$. Then for all $\psi,\phi\in D$ the map
	$$ f(t) = \braket{\psi,e^{-itB}Ae^{itB}\phi} $$
	is continuously differentiable with derivative
	$$ f'(t) = iq_{[A,B]}(e^{itB}\psi,e^{itB}\phi). $$
\end{lem}
\begin{proof}
	We easily calculate
	$$ f(t+h)-f(t) = \braket{(e^{ihB}-1)e^{itB}\psi,Ae^{i(t+h)B}\phi}+\braket{Ae^{itB}\psi,(e^{ihB}-1)e^{itB}\phi}.$$
	The statement then directly follows using the continuity assumption.
\end{proof}

\section[Transformation and Convergence Properties of Weyl Operators]{Transformation and Convergence Properties of\\ Weyl Operators}

In this appendix we prove transformation and convergence properties of Weyl operators defined in \cref{defn:Weyloperator}, which we need in the proof of \cref{LargeFiberProp}. Throughout this appendix we assume $\omega:\IR^\nu\to\IR$ is measurable and satisfies $\omega>0$ almost everywhere and $h:\IR^{\nu}\rightarrow \IR^{p}$ is measurable. We define the selfadjoint operators
$$ B:=(1+d\Gamma(\omega))^{1/2}\quad\mbox{and}\quad A_s:=1+d\Gamma(\omega)^{1/2}+|d\Gamma(h)|^s \quad\mbox{for}\ s\in[0,1]. $$
Note that $\cD(A_s)=\cD(B)\cap \cD(|d\Gamma(h)|^s)$ and $\cD(B)=\cD(d\Gamma(\omega)^{1/2})$. 

\begin{thm}\label{thm:Weyltransformation-strong}
	Let  $f\in \cD(\omega^{1/2})$ and $s\in[0,1]$. Then:
	\begin{enumerate}[(1)]
		\item \label{part:Weylunbounded} $W(f)\cD(B)=\cD(B)$ and
		$
		\|(d\Gamma(\omega)+1)^{1/2}W(f)(d\Gamma(\omega)+1)^{-1/2}\| \le 1+\|\omega^{1/2}f\|.
		$
		
		\item\label{part:squareintegrable} If $f\in\cD(\lvert h\lvert)\cap \cD(\omega^{-1/2}\lvert h\lvert )$, then $W(f)\cD(A_1)=\cD(A_1)$ and on $\cD(A_1)$ we have
		\begin{equation}\label{eq:transformationWeyl}
		W(f) d\Gamma(h_i) W(f)^*=d\Gamma(h_i)-\ph(h_if)+\langle f,h_if \rangle \qquad\mbox{for}\ i\in\{1,\ldots,p\}.
		\end{equation}
		
		\item\label{part:notsquareint} If $f\in\cD(|h|^{1/2})\cap\cD(|h|\omega^{-1/2})\setminus\cD(h_i)$, then $\cD(d\Gamma(h_i)) \cap W(f)^*\cD(A_1) = \{0\}$ for all $i\in\{1,\ldots,p\}$
		
		\item\label{part:tothepowers} If $f\in\cD(|h|^s)\cap\cD(\omega^{-1/2}|h|^s) $, then $W(f)\cD(A_s) = \cD(A_s)$. Furthermore, if $(f_n)\subset \cD(|h|^s)\cap \cD(\omega^{-1/2}|h|^s)$ converges to $f$ simultaneously in $\omega^{1/2}$-, $|h|^s$- and $\omega^{-1/2}|h|^s$-norm, then
		\begin{equation}\label{eq:convergenceWeyl}
		\slim_{n\to\infty} A_sW(f_n)A_s^{-1}  = A_sW(f_n)A_s^{-1}.
		\end{equation}
	\end{enumerate}
\end{thm}
\noindent
This theorem is an extension of \cite[Lemma C.3, C.4 \& Cor. C.5]{GriesemerWuensch.2018} and \cite[Lemma A.4]{HiroshimaMatte.2019}. The proof is obtained in several lemmas. For $\eps>0$ and $\lambda\in[0,\infty)$ we write
$$ B_\eps (\lambda) := (\eps d\Gamma(\omega)+1+\eps\lambda)^{1/2}\quad\mbox{and}\quad B_\eps:=B_\eps(0).  $$
Further, if $f$ satisfies $W(f)\cD(B)\subset \cD(B)$, then we define
$$C_{f,\varepsilon}=B_\varepsilon W(f)B_\varepsilon^{-1} \quad\mbox{and}\quad C_{f}:=C_{f,1}.$$

\begin{lem}\label{core}
	The set $\cE := \{ \epsilon(g) \mid g\in \cD(\omega+\lvert h \lvert) \}$ spans a core for $A_s$. 
\end{lem}
\begin{proof}
	Note that $\cE$ spans a dense set inside $\cD(d\Gamma(\omega+\lvert h \lvert))$ (cf. \cite[Prop. 20.7]{Parthasarathy.1992}) and is left invariant by $\Gamma(e^{it(\omega+\lvert h \lvert)})=e^{itd\Gamma(\omega+\lvert h \lvert)}$, so $\cE$ spans a core for $d\Gamma(\omega+\lvert h \lvert)$ (\cite[Thm. VIII.11]{ReedSimon.1972}). Now $A_s$ commutes with $d\Gamma(\omega+\lvert h \lvert)$ and is $d\Gamma(\omega+|h|)$-bounded by \cref{Lem:propertiesOfFuncDGamma}, so the conclusion follows from \cref{commuteingCore}
\end{proof}

\begin{lem}\label{lem:almostThere}
	If $f\in\cD(\lvert h\lvert)\cap \cD(\omega^{-1/2}\lvert h\lvert )$, then $W(f)\cD(A_1) \subset \cD(d\Gamma(h_i))$ for all $i\in \{1,\ldots,p \}$ and \cref{eq:transformationWeyl} holds on $\cD(A_1)$.
\end{lem}
\begin{proof}
	The left hand side of \cref{eq:transformationWeyl} is a closed operator and the right hand side is $A_1$-bounded, so by \cref{core} it suffices to show \cref{eq:transformationWeyl} holds on $\cE$.
	Let $g\in\cD(h_i)$. Then using \cref{defn:dGamma-scalar,defn:creation,defn:corherentstate} we easily see $\epsilon(g)\in\cD(d\Gamma(h))\cap \cD(a^\dag(hg))$ and
	$d\Gamma(h_i)\epsilon(g)=a^\dag(h_ig)\epsilon(g).$
	Further, \cref{defn:annihilation,defn:corherentstate} yield $a(f)\epsilon(g) = \braket{f,g}\epsilon(g)$ for all  $f,g\in \HS$.
	Now let $g_1,g_2\in\cD(\omega+|h|)$. Using \cref{defn:Weyloperator}, $W(f)^* = W(-f)$ and $\braket{\epsilon(g_2),\epsilon(g_1)}=e^{\braket{g_2,g_1}}$ we then have
	\begin{align*}
	\braket{\epsilon(g_2),W(f)d\Gamma(h)W(f)^*\epsilon(g_1)}
	&= e^{-\|f\|^2+\braket{f,g_1}+\braket{g_2,f}}\braket{\epsilon(g_2-f),d\Gamma(h)\epsilon(g_1-f)}\\
	&= \braket{g_2-f,h(g_1-f)}e^{\braket{g_2,g_1}}\\
	&= \braket{\epsilon(g_2),(d\Gamma(h)-\ph(hf)+\braket{f,hf})\epsilon(g_1)}
	\end{align*}
	As $\cE$ is total (as it spans a core), this proves \cref{eq:transformationWeyl} holds on $\cE$.
\end{proof}
\noindent 

\begin{lem}\label{ConvGeneral}
	Let $\eps>0, (f_n)\subset L^2(\IR^\nu)$ and assume $\lim\limits_{n\to\infty} f_n=: f$ exists in $L^2(\IR^\nu)$.\\ If $W(f_n)\cD(B)\subset \cD(B)$ for all $n\in\IN$ and both $C_{f_n,\eps}^*C_{f_n,\eps}$  and $C_{-f_n,\eps}^*C_{-f_n,\eps}$ converge strongly, then $W(f)\cD(B)\subset \cD(B)$ and $\slim\limits_{n\to\infty} C_{f_n,\eps}=C_{f,\eps}$, $\slim\limits_{n\to\infty} C_{f_n,\eps}^*=C_{f,\eps}^*$.
\end{lem}
\begin{proof}
	Since $\lVert C_{\pm f_n,\eps}^*C_{\pm f_n,\eps} \lVert=\lVert C_{\pm f_n,\eps}^* \lVert^2$, we see $C_{\pm f_n,\eps}^*$ is uniformly bounded. As
	$\slim\limits_{n\rightarrow \infty} B_\eps^{-1}W(f_n)=B_\eps^{-1}W( f)$ by \cref{Prop:Bregninger}, \partref{WeakStrongConvergence}{part:strongconvergenceproduct} implies
	$W(f)\cD(B)\subset \cD(B)$ and $\slim\limits_{n\rightarrow \infty} C_{\pm f_n,\eps}^*=C_{\pm f,\eps}^*$.
	Taking adjoints, this yields $\wlim\limits_{n\rightarrow \infty} C_{f_n,\eps} = C_{f,\eps}$. Now, the equality $C_{f_n,\eps} = C_{- f_n,\eps}^*(C_{f_n,\eps}^*C_{ f_n,\eps})$ shows the convergence is actually strong.
\end{proof}
In the following \lcnamecref{convOdd}, we discuss the construction of field operators for not square-integrable functions.
\begin{lem}\label{convOdd}
	Let $f,g:\RR^\nu\to \CC$ satisfy $\omega^{-1/2}f,\omega^{-1/2}g\in L^2(\RR^\nu)$. Then for $\eps>0$
	\begin{enumerate}[(1)]
		\item\label{part:l2sequence}  There is $(f_n)\subset\cD(\omega^{-1/2})$ such that $\lim\limits_{n\rightarrow \infty }\omega^{-1/2}f_n=\omega^{-1/2}f$ in $L^2(\RR^\nu)$.
		\item If $(f_n)$ is a sequence as in \cref{part:l2sequence}, then there is a bounded operator $\widetilde\ph_\eps(f)$ independent of the sequence such that $\widetilde \ph_\eps(f)=\lim\limits_{n\to\infty} B_\eps^{-1}\ph(f_n)B_\eps^{-1} $.
		\item $\lVert \widetilde{\varphi}_\eps(f) \lVert\leq 2\eps^{-1/2}\lVert  \omega^{-1/2}f \lVert$ and $\widetilde{\varphi}_\eps(f)-\widetilde {\varphi}_\eps(g)=\widetilde{\varphi}_\eps(f-g)$.
	\end{enumerate}
\end{lem}
\begin{proof}
	A possible choice in \cref{part:l2sequence} is $f_n(k)=f(k)1_{\{ \lvert f\lvert<n, n^{-1}<\omega,|k|<n  \}}(k)$.
	\cref{Prop:Bregninger} (applied with $\varepsilon \omega$ instead of $\omega$) yields $\|a(h)B_\eps^{-1}\|\le \eps^{-1/2} \|\omega^{-1/2}h\| $ for $h\in\cD(\omega^{-1/2})$. Hence, using $\|B_\eps^{-1}a^\dag(h)\|=\|(a(h)B_\eps^{-1})^*\|$ and $B_\eps\ge 1$, we find $$ \|B_\eps^{-1}\ph(h)B_\eps^{-1}\|\le 2 \eps^{-1/2} \|\omega^{-1/2}h\|.$$
	This inequality, the closedness of the bounded operators and the fact that for all $h_1,h_2\in\cD(\omega^{-1/2})$ the equality $B_\eps^{-1}\ph(h_1)B_\eps^{-1}-B_\eps^{-1}\ph(h_2)B_\eps^{-1}=B_\eps^{-1}\ph(h_1-h_2)B_\eps^{-1}$ holds  then finishes the proof.
\end{proof}

\begin{lem}\label{ConvMatte}
	Let $\eps>0$ and $f\in \cD(\omega^{1/2})$. Then $W(f)\cD(B)\subset \cD(B)$. Further, if $(f_n)\subset \cD(\omega^{1/2})$ converges to $f$ in $\omega^{1/2}$-norm, then $\slim\limits_{n\rightarrow \infty} C_{f_n,\eps}=C_{f,\eps}$,  $\slim\limits_{n\rightarrow \infty} C_{f_n,\eps}^*=C_{f,\eps}^*$ and $\lVert C_{f} \lVert\leq 1+\lVert  \omega^{1/2}f \lVert$.
\end{lem}
\begin{proof}
	We set $g_n=f 1_{\{\omega<n\}} \in \cD(\omega)$. We apply \cref{lem:almostThere} with $p=1$ and $h=\omega$. This yields $W(g_n)\cD(B^2)\subset \cD(B^2)$, which implies $W(g_n)\cD(B)\subset \cD(B)$ by \cref{lem:domainsqrt}. Further, using \cref{eq:transformationWeyl},  on $\cD(B)$
	\begin{equation*}
	C_{g_n,\eps}^*C_{g_n,\eps}= B_\eps^{-1}W(g_n)^*B_\eps^2W(g_n)B_\eps^{-1}= 1-\eps B_\eps^{-1}{\varphi}(\omega f_n)B_\eps^{-1}+\eps\lVert \omega^{1/2}g_n \lVert^2 B_\eps^{-2}.
	\end{equation*}
	By \cref{convOdd}, the right hand side converges in norm as $n\to\infty$, so \cref{ConvGeneral} shows $W(f)\cD(B)\subset \cD(B)$ and 
	\begin{equation}\label{eqgtig}
	C_{f,\eps}^*C_{f,\eps}=1 -\eps\widetilde{\varphi}_\eps(\omega f)+\eps\lVert \omega^{1/2}f \lVert^2 B_\eps^{-2} \qquad \mbox{for any}\ f\in \cD(\omega^{1/2}).
	\end{equation}
	Another application of \cref{convOdd} then shows $C_{f_n,\eps}^*C_{f_n,\eps}$ is convergent if $(f_n)\subset \cD(\omega^{1/2})$ converges to $f$ in $\omega^{1/2}$-norm, and hence \cref{ConvGeneral} shows $\slim\limits_{n\rightarrow \infty} C_{f_n,\eps}=C_{f,\eps}$ and  $\slim\limits_{n\rightarrow \infty} C_{f_n,\eps}^*=C_{f,\eps}^*$. \cref{eqgtig,convOdd} now imply  $\lVert C_{f} \lVert^2= \lVert C_{f}^*C_{f} \lVert\leq 1+2\lVert \omega^{1/2}f \lVert+\lVert \omega^{1/2}f \lVert^2$.
\end{proof}

\begin{lem}\label{bound}
	Let $i\in \{1,..,p\}$ and $f\in\cD(|h_i|^{1/2})\cap \cD(\omega^{1/2})\cap \cD(|h_i|\omega^{-1/2})$.\\ Set $f_\Lambda=f1_{\{ \lvert h\lvert<\Lambda \}}$. 
	If $\psi\in \cD(B)\cap\cD(d\Gamma(h_i))$ and $W(f)\psi\in \cD(d\Gamma(h_i))$, then $$\limsup_{\varepsilon \rightarrow  0 } \limsup_{\Lambda\rightarrow  \infty } \lVert B_\varepsilon^{-1} a^{\dagger}(h_i f_\Lambda)\psi\lVert<\infty.$$
\end{lem}
\begin{proof}
	By \cref{lem:almostThere}, we have
	\begin{equation}\label{eq:core223}
	B_\eps^{-1}W(f_\Lambda)d\Gamma(h_i)W(f_\Lambda)^*\psi = d\Gamma(h_i)B_\varepsilon^{-1}\psi-\widetilde\ph(h_i f_\Lambda) B_\eps\psi+\braket{f_\Lambda,h_if_\Lambda}B_\varepsilon^{-1}\psi.
	\end{equation}
	Setting $C=\|d\Gamma(h_i)\psi\|+\||h_i|\omega^{-1/2}f\|\|d\Gamma(\omega)^{1/2}\psi\|+\||h_i|^{1/2}f\|^2\|\psi\|$, \cref{Prop:Bregninger} and \cref{eq:core223} imply
	$\|B_\varepsilon^{-1} a^{\dagger}(h_i f_\Lambda)\psi\|\le \|B_\eps^{-1}W(f_\Lambda)d\Gamma(h_i)W(f_\Lambda)^*\psi\|+C. $ 
	Note that \cref{convOdd} implies the right hand side of \cref{eq:core223} converges as $\Lambda\to\infty$. Hence,
	$$ B_\varepsilon^{-1}d\Gamma(h_i)W(f_\Lambda)^*\psi = C_{f_\Lambda,\varepsilon}^*  B_\eps^{-1}W(f_\Lambda)d\Gamma(h_i)W(f_\Lambda)^* \psi$$
	must also converge by \cref{ConvMatte}. As $B_\varepsilon^{-1}d\Gamma(h_i)$ is closed, the limit must be $B_\varepsilon^{-1}d\Gamma(h_i)W(f)^*\psi$. Hence, we obtain
	$$
	B_\eps^{-1}W(f_\Lambda)d\Gamma(h_i)W(f_\Lambda)^*\psi = C_{-f_\Lambda,\varepsilon}^*B_\varepsilon^{-1}d\Gamma(h_i)W(f_\Lambda)^*\psi
	$$
	converges to $B_\eps^{-1}W(f)d\Gamma(h_i)W(f)^*\psi$ as $\Lambda\to\infty$. We obtain
	\begin{equation*}
	\limsup_{\varepsilon \rightarrow  0 } \limsup_{\Lambda\rightarrow  \infty } \lVert B_\varepsilon^{-1}W(f_\Lambda )d\Gamma(h_i)W(f_\Lambda)^*\psi \lVert=\lVert d\Gamma(h_i)W(f)^*\psi \lVert<\infty.\qedhere
	\end{equation*}
\end{proof}
\begin{lem}\label{lem:zerovector}
	Let $i\in \{1,..,p\}$ and $f\in\cD(h_i\omega^{-1/2})\setminus\cD(h_i)$. Set
	$f_\Lambda=f1_{\{ \lvert h\lvert<\Lambda\}}$.  
	If $\psi\in \cD(B)$ and $\limsup\limits_{\varepsilon \rightarrow  0 } \limsup\limits_{\Lambda\rightarrow  \infty } \lVert B_\varepsilon^{-1} a^{\dagger}(h_i f_\Lambda)\psi\lVert<\infty$, then $\psi=0$.
\end{lem}
\begin{proof}
	We use the definitions \cref{defn:annihilation,defn:creation} and obtain
	\begin{align*}
	(a(h_if_\Lambda) B^{-2}_\varepsilon a^\dagger(h_i&f_\Lambda)\psi^{(n)}) (k_1,...,k_n) = \int_{ \RR^\nu}\frac{\lvert h_i(k)f_\Lambda(k) \lvert^2 \psi(k_1,...,k_n)}{1+\varepsilon\omega(k)+\varepsilon\omega^{(n)}(k_1,...,k_n) }dk
	\\&
	+\sum_{j=1}^{n} h_i(k_j)f_\Lambda(k_j) \int_{ \RR^\nu}\frac{\overline{h_i(k)f_\Lambda(k)}\psi^{(n)}(k,k_1,...,\hat{k}_j,...,k_n ) }{1+\varepsilon\omega^{(n+1)}(k,k_1,...,k_n) }dk.
	\end{align*}
	The second term is bounded by $a^\dag(|h_if_\Lambda|)a(|h_if_\Lambda|)\lvert \psi^{(n) }\lvert(k_1,\ldots,k_n)$, since $\omega\ge0$. Hence, we obtain
	$$
	\lVert B_\varepsilon^{-1} a^{\dagger}(h_i f_\Lambda)\psi^{(n)}\lVert^2\geq \int_{\IR^\nu} \lvert h_i(k)f_\Lambda(k) \lvert^2\lVert B_\varepsilon(\omega(k))\psi^{(n)} \lVert^2dk-\lVert a(\lvert h_if_\Lambda\lvert ) \lvert \psi^{(n)} \lvert \lVert^2
	$$
	By \cref{Prop:Bregninger}, $\lVert a(\lvert h_if_\Lambda\lvert ) \lvert \psi^{(n)}\lvert \lVert^2\leq \lVert h_i \omega^{-1/2} f  \lVert^2 \lVert B\psi^{(n)} \lVert^2$, so summing over $n$ and using monotone convergence in the limits $\Lambda\to\infty$ and $\varepsilon\to\infty$ we get
	$$
	\infty > \limsup_{\Lambda\to\infty}\int_{\IR^\nu} \lvert h_i(k)f_\Lambda(k) \lvert^2dk \lVert \psi \lVert^2-\lVert h_i \omega^{-1/2} f  \lVert^2 \lVert B\psi \lVert^2.
	$$
	Since $h_if$ is not square-integrable, this implies $\|\psi\|=0$.
\end{proof}
\begin{lem}\label{lem:commutatorbound}
	Let $s\in [0,1]$ and $f\in \cD(\omega^{1/2})\cap \cD(\lvert h\lvert^s)\cap(\omega^{-1/2}\lvert h\lvert^s)$.\\
	Then there is a unique bounded operator $D_{f,s}$ such that (cf. \cref{defn:commutatorform})
	$$ q_{[|d\Gamma(h)|^s,\ph(f)]}(\psi,B^{-1}\phi)=\braket{\psi, D_{f,s}\phi} \quad\mbox{for}\ \psi,\phi\in\cD(A_s). $$
	If $(f_n)\subset \cD(\omega^{-1/2})\cap \cD(\omega^{-1/2}\lvert h\lvert^s)\cap \cD(\lvert h\lvert^s)$ converges to $f$ in $\omega^{-1/2},\omega^{-1/2}\lvert h\lvert^s$ and $\lvert h\lvert^s$ norm, then $\lim\limits_{n\to\infty} D_{f_n,s} = D_{f,s}$. Further, $ \lVert D_{f,s} \lVert\leq 2\lVert (1+\omega^{-1/2})\lvert h\lvert^s f \lVert$.
\end{lem}
\begin{proof}
	Let $g\in \cD(\lvert h^{(n)}\lvert^s)$ (cf. \cref{eq:defining_eqn}). We use $\lvert \lvert x+y  \lvert^s-\lvert x  \lvert^s \lvert\leq \lvert y\lvert^s$ for all $x,y\in \RR^\nu$ and the definition of the anihilation operator (cf. \cref{defn:annihilation}) and see the inequality
	\begin{align*}
	\lvert a(f)\lvert h^{(n)}\lvert^sg- \lvert h^{(n-1)}\lvert^sa(f)g \lvert\le a(\lvert h\lvert^s\lvert f\lvert)\lvert g \lvert
	\end{align*}
	holds pointwise. Now let $\phi,\psi\in \cD(A_s)$ and define $\lvert \psi \lvert=\{ \lvert \psi^{(n)} \lvert \},\lvert \phi \lvert=\{ \lvert \phi^{(n)} \lvert \}\in \cD(A_s)$. We now write $q_{f,s}=q_{[|d\Gamma(h)|^s,\ph(f)]}$ and use the above inequality to obtain
	\begin{align*}
	\lvert q_{f,s}(\psi,\phi)\lvert\leq \langle \lvert \psi\lvert,a(\lvert h\lvert^s\lvert f\lvert)\lvert \phi\lvert \rangle+\langle a(\lvert h\lvert^s\lvert f\lvert)\lvert \psi\lvert,\lvert \phi\lvert \rangle=\langle \lvert \psi\lvert,\varphi(\lvert h\lvert^s\lvert f\lvert)\lvert \phi\lvert \rangle.
	\end{align*}
	Since $\omega\ge 0$, we have $\lvert B^{-1} \phi\lvert=B^{-1}\lvert  \phi\lvert$. Replacing $\phi$ by $B^{-1}\phi$ and using \cref{Prop:Bregninger} as well as the Cauchy-Schwarz inequality yields
	\begin{align*}
	\lvert q_{f,s}(\psi,B^{-1}\phi)\lvert \leq 2\lVert (1+\omega^{-1/2})\lvert h\lvert^s f \lVert \lVert \psi\lVert \lVert \phi \lVert.
	\end{align*}
	This proves existence and the upper bound on $D_{f,s}$.
	
	We observe $q_{f,s}-q_{f_n,s}=q_{f-f_n,s}$, which yields $D_{f,s}-D_{f_n,s}=D_{f-f_n,s}$ by the bound above. The convergence statement directly follows.
\end{proof}
\begin{proof}[{\bf Proof of \cref{thm:Weyltransformation-strong}}]
	Recalling $W(f)^*=W(-f)$ \cref{part:Weylunbounded,part:squareintegrable} follow from \cref{ConvMatte,lem:almostThere}. Further, \cref{bound,lem:zerovector} yield \cref{part:notsquareint}.
	
	Hence, it remains to prove \cref{part:tothepowers}. Therefore, let $f_\Lambda(k) = 1_{\{|h|\le \Lambda\}}f(k)$ and recall that by \cref{part:squareintegrable} $W(tf_\Lambda)$ maps $\cD(A_1)$ onto itself for all $t\in\IR$. Now assume $\phi,\psi\in\cD(A_1)$ and define
	\begin{equation}\label{eq:g-defn}
	g_{\Lambda,\psi,\phi} (t) = \braket{W(tf_\Lambda)\psi,|d\Gamma(h)|^sW(tf_\Lambda)\phi} \qquad\mbox{for}\ t\in\IR.
	\end{equation}
	For all $i\in\{1,\ldots,p\}$ the map
	$$
	t\mapsto d\Gamma(h_i) W(tf_\Lambda)\psi=W(tf_\Lambda)( d\Gamma(h_i)\psi +t\varphi(f_\Lambda)\psi +t^2\langle h_if_\Lambda,f_\Lambda \rangle \psi )
	$$
	is continuous by \cref{Prop:Bregninger}, so $W(tf_\Lambda)\psi$ is continuous in $d\Gamma(h_i)$-norm. Since
	$$
	\lVert \lvert d\Gamma(h)\lvert^s \eta\lVert \leq \lVert \eta\lVert+\sum_{i=1}^{n}\lVert d\Gamma(h_i)\eta\lVert \qquad \mbox{for all}\  \eta\in \cD(\lvert d\Gamma(h)\lvert) 
	$$
	by the spectral theorem, $t\mapsto \lvert d\Gamma(h)\lvert^s W(tf_\Lambda)\psi$ is continuous, so we can apply \cref{lem:commutator} to \cref{eq:g-defn}. Hence, $g_{\Lambda,\psi,\phi}$ is continuously differentiable with derivative
	$$ g_{\Lambda,\psi,\phi}'(t) = -iq_{\Lambda}(W(tf_\Lambda)\psi,W(tf_\Lambda)\phi), \quad\mbox{where}\ q_\Lambda = q_{[|d\Gamma(h)|^s,\ph(f_\Lambda)]}\ \mbox{as in \cref{defn:commutatorform}}. $$
	By \cref{lem:commutatorbound}, the form $q_\Lambda(\psi,B^{-1}\psi)$ corresponds to an operator $D_\Lambda\in\cB(\FS)$ bounded uniformly in $\Lambda$ and satisfying $\lim\limits_{\Lambda\to\infty} D_\Lambda = D_\infty$. We now therefore have
	\begin{align*}
	&\braket{\psi,|d\Gamma(h)|^sW(f_\Lambda)\phi} = g_{\Lambda,W(-f_\Lambda)\psi,\phi}(1)\\
	&\qquad= g_{\Lambda,W(-f_\Lambda)\psi,\phi}(0) + \int_0^1 g_{\Lambda,W(-f_\Lambda)\psi,\phi}' (t)dt\\
	&\qquad=\braket{\psi,W(f_\Lambda)|d\Gamma(h)|^s\phi} -i \int_0^1\braket{\psi,W((1-t)f_\Lambda)D_\Lambda BW(tf_\Lambda)\phi}dt.
	\end{align*}
	Since $\psi\in \cD(A_1)$ was arbitrary and $\cD(A_1)$ is dense, this yields
	$$ |d\Gamma(h)|^sW(f_\Lambda)\phi = W(f_\Lambda)|d\Gamma(h)|^s\phi -i\int_0^1 W((1-t)f_\Lambda)D_\Lambda BW(tf_\Lambda)\phi dt, $$
	where we use the Bochner integral on the right hand side. By the dominated convergence theorem, \cref{ConvMatte,Prop:Bregninger} we can take the limit $\Lambda\to\infty$ and obtain
	$ W(f)\phi\in \cD(|d\Gamma(h)|^s) $ as well as
	\begin{equation}\label{intFOrm}
	|d\Gamma(h)|^sW(f)\phi = W(f)|d\Gamma(h)|^s\phi -i\int_0^1 W((1-t)f)D_\infty BW(tf) \phi dt. 
	\end{equation}
	Applying the dominated convergence theorem once more, \cref{intFOrm} also holds for all $\phi\in\cD(A_s)$, so 
	$W(-f)=W(f)^*$ directly yields $W(f)\cD(A_s)=\cD(A_s)$.
	Finally, \cref{eq:convergenceWeyl} follows by \cref{ConvMatte} and a dominated convergence type argument using \cref{lem:commutatorbound} applied to \cref{intFOrm} as before. 
\end{proof}

\section{Decomposable Hilbert Space Operators}
\label{app:directintegrals}

In this appendix we introduce the notation of direct integrals. The definition and a few simple lemmas are necessary to define our model and for the fiber decomposition. Further, we prove some convergence statements, which are essential in the proof of \cref{LargeFiberProp}.

Assume $\fH$ is a separable Hilbert space. Then $L^2(\IR^\nu)\otimes\fH$ is the vector valued $L^2$-space $L^2(\IR^\nu,\fH)$ with $f\otimes \psi$ being the map $x\mapsto f(x)\psi$. We call a map $f:\IR^\nu\to\cB(\fH)$ strongly measurable, if $x\mapsto f(x)\psi$ is measurable for all $\psi\in\fH$. Further, if $x\mapsto \|f(x)\|$ is essentially bounded, we define the direct integral $I_{\oplus,x}(f(x))$ to be the bounded operator on $L^2(\IR^\nu,\fH)$ defined by
\begin{align*}
I_{\oplus,x}(f(x))\psi:=\int_{\cM}^{\oplus} f(x) dx \psi : x \mapsto f(x)\psi(x) \qquad\mbox{for}\ \psi\in L^2(\IR^\nu,\fH).
\end{align*}
One may prove that the norm is given by the essential supremum (see \cite[Theorem XIII.83]{ReedSimon.1978}), i.e.,
\begin{align}\label{eq:esssupnorm}
\left\|I_{\oplus,x} (f(x))\right\|= \esssup_{x\in \IR^\nu}\|f(x)\|.
\end{align}
We call a collection of selfadjoint operators $\{A_x\}_{x\in\IR^\nu}$ strong resolvent measurable if $x\mapsto (A_x+i)^{-1}$ is strongly measurable. Then we define the direct integral
$$ I_{\oplus,x}(A_x)\psi=\int_{\IR^\nu}^\oplus A_xdx\psi : x\mapsto  A_x\psi(x) $$
with domain
\begin{align*}
\cD(I_{\oplus,x}(A_x)) = 
\left\{ \psi \in L^{2}(\IR^\nu,\fH) \biggl \lvert \psi(x) \in \cD(A_x)\text{ for almost all $x$ and }\int_{\IR^\nu} \|A_x\psi(x)\|^2dx<\infty   \right\}.
\end{align*}
The following theorem sums up the results about direct integrals we shall need.
\begin{lem}[based on {\cite[Theorem XIII.85]{ReedSimon.1978}}]\label{Thm:Dirint}
	The following holds
	\begin{enumerate}[(1)]
		\item\label{measureAble} A collection of selfadjoint operators $\{ A_x\}_{x\in \IR^\nu} $ is strong resolvent measurable if and only if $x\mapsto e^{itA_x}$ is weakly measurable.\\ In this case, if $\psi:\IR^\nu\rightarrow \fH$ is measurable and $\psi(x)\in \cD(A_x)$ for all $x\in\IR^\nu$, then $x\mapsto A_x\psi(x)$ is measurable.
		\item\label{funcCalc} Let $\{A_x\}_{x\in\IR^\nu}$ be a strong resolvent measurable collection of selfadjoint operators. Then $I_{\oplus,x}(A_x)$ is selfadjoint.\\ Further, if $f:\IR\to\IR$ is measurable, then $\{f(A_x)\}_{x\in \IR^\nu}$ is strong resolvent measurable and $f(I_{\oplus,x}(A_x))=I_{\oplus,x}(f(A_x))$. If there is $\lambda\in \RR$ such that $A_x\geq \lambda$ for all $x\in\IR^\nu$, then $I_{\oplus,x}(A_x)\geq \lambda$.
		\item\label{dirInt} If $A$ is selfadjoint or bounded on $\fH$, we may identify $1\otimes A=I_{\oplus,x}(A)$.
	\end{enumerate}
\end{lem}
\noindent
The next lemma concerns direct integrals of Fock space operators, as introduced in \cref{sec:notation}, i.e., $\fH=\cF$.
\begin{lem}\label{Prop:Directint}
	Assume $f\in L^2(\RR^\nu)$ is measurable and let $\{U_x\}_{x\in\cM}$ be a strongly measurable family of unitary operators on $L^2(\RR^\nu)$. Then
	\begin{enumerate}[(1)]
		\item $\{ \varphi(U_xf)  \}_{x\in \IR^\nu}$ is strong resolvent measurable.\\ Further, $x\mapsto W(U_xf)$ and $x\mapsto \Gamma(U_x)$ are strongly measurable.
		\item Let $\omega:\IR^\nu\to\IR$ be a multiplication operator with $\omega> 0$ almost everywhere. If $f\in \cD(\omega^{-1/2})$ then $I_{\oplus,x}(\varphi(U_xf))$ is infinitesimally $1\otimes d\Gamma(\omega)$-bounded and $I_{\oplus,x}(\varphi(U_xf))+1\otimes d\Gamma(\omega)\geq \|\omega^{-1/2}f\|^2$.
	\end{enumerate}
\end{lem}
\begin{proof}
	Combine \cref{Thm:Dirint,Prop:Bregninger}.
\end{proof}
\noindent
The next two statements connect norm convergence and norm resolvent convergence with the direct integral decomposition of bounded and selfadjoint operators, respectively. This is essential in the proof of \cref{LargeFiberProp}.
\begin{lem}\label{Lem:conccc}
	For any $\Lambda>0$, let $f_\Lambda:\IR^\nu\to \cB(\fH)$ be continuous and bounded and set $B_\Lambda=I_{\oplus,x}(f_\Lambda(x))$.
	Assume $B_{\Lambda}$ converges to an operator $B$ in norm as $\Lambda\to \infty$. Then there is a continuous and bounded function $f:\IR^\nu\to\cB(\fH)$ such that $f_{\Lambda}$ uniformly converges to $f$ as $\Lambda\to\infty$ and $B=I_{\oplus,x}(f(x))$.
\end{lem}
\begin{proof}
	Let $\Lambda,\Lambda'>0$. Since $f_\Lambda$ and $f_{\Lambda'}$ are continuous, \cref{eq:esssupnorm} implies
	$$ \|B_\Lambda-B_{\Lambda'}\| = \esssup_{x\in \RR^\nu}\|f_\Lambda(x)-f_{\Lambda'}(x)\| = \sup_{x\in \RR^\nu}\|f_\Lambda(x)-f_{\Lambda'}(x)\|.$$
	Hence, $(f_{\Lambda})_{\Lambda>0}$ is a Cauchy net in the Banach space $C_b(\RR^\nu,\cB(\fH))$ and there is a limit $f\in C_b(\RR^\nu,\cB(\fH))$. Clearly, $B_\Lambda=I_{\oplus,x}(f_\Lambda(x))$ converges to $I_{\oplus,x}(f(x))$, showing $B=I_{\oplus,x}(f(x))$. Let $\varepsilon>0$ and pick $\Lambda_0>0$ such that $\|B-B_\Lambda\|<\varepsilon$ for all $\Lambda>\Lambda_0$. Then $\sup_{x\in \RR^\nu}\|f(x)-f_\Lambda(x)\|<\varepsilon$ follows by continuity, finishing the proof.
\end{proof}
\noindent
A collection $\{A_x\}_{x\in\IR^\nu}$ of selfadjoint operators is called resolvent continuous, if it is uniformly bounded below and for some $z<\inf_{x\in \RR^\nu}\inf(\sigma(A_x))  $ the map $x\mapsto (A_x-z)^{-1}$ is continuous in norm. Note this implies $x\mapsto \inf\sigma(A_x)$ is continuous by \cref{lem:normresolventConv}.
\begin{lem}\label{Lem:ConvergenceUniform}
	For any $\Lambda>0$, let $\{A_{\Lambda,x}\}_{x\in\IR^\nu}$ be a resolvent continuous collection of selfadjoint operators on $\fH$ and define $A_\Lambda=I_{\oplus,x}(A_{\Lambda,x})$. Assume $A_\Lambda$ is bounded below by $\lambda$ uniformly in $\Lambda>0$ and converges to a selfadjoint operator $A$ in the norm resolvent sense as $\Lambda\to\infty$. Then $\inf\sigma(A_{\Lambda,x})\geq \lambda$ for all $(x,\Lambda)\in\IR^\nu\times(0,\infty)$ and there is a strong resolvent measurable family of selfadjoint operators $\{A_x\}_{x\in\IR^\nu}$ such that $A=I_{\oplus,x}(A_x)$ and $A_x\geq \lambda$ for all $x\in \RR^\nu$. Further, there exists a nullset $\cN\subset \big\{ x\in \RR^\nu \big| A_{\Lambda,x}\text{ has no strong resolvent limit as }\Lambda\to \infty \big\}$ such that $A_{\Lambda,x}$ converges to $A_x$ in the norm resolvent sense if $x\in\cN^c$ and $\{A_x\}_{x\in\IR^\nu}$ is resolvent continuous in $\cN^c$.
\end{lem}
\begin{proof}
	By \cref{lem:normresolventConv}, $x\mapsto \inf(\sigma(A_{\Lambda,x}))$ is continuous for each $\Lambda\in (0,\infty)$, so \begin{align*}
	e^{-\lambda}\geq  \sup_\Lambda \lVert e^{-A_n} \lVert=\sup_\Lambda \esssup_x \lVert e^{-A_{\Lambda,x}} \lVert=\sup_{(\Lambda,x)}e^{-\inf(\sigma(A_{\Lambda,x}))}.
	\end{align*}
	This proves $\inf(\sigma(A_{\Lambda,x}))\geq \lambda$ for all $x\in \RR^\nu$ and $\Lambda \in \RR$.
	
	Let $z<\lambda$. Note $x\mapsto (A_{\Lambda,x}-z)^{-1}$ is continuous by \cref{lem:normresolventConv}, so \cref{Lem:conccc} implies there is a continuous map $x\mapsto B_x$ to which $x\mapsto (A_{\Lambda,x}-z)^{-1}$ converges uniformly as $\Lambda\to\infty$. It follows that $B_x$ is selfadjoint and $0\leq B_x\leq \frac{1}{\lambda-z}$, since $(A_{\Lambda,x}-z)^{-1}$ is selfadjoint and $0\leq (A_{\Lambda,x}-z)^{-1}\leq \frac{1}{\lambda-z}$ for all $x\in \IR^\nu$ and $\Lambda\in (0,\infty)$. Further, we have $I_{\oplus,x}( B_x  )=(A-z)^{-1}$, so \cref{Thm:Dirint} yields
	\begin{equation*}
	I_{\oplus,x}( 1_{\{ 0 \}}(B_x)  )=1_{\{0\}} ((A-z)^{-1})=0
	\end{equation*}
	as $A-z$ is injective. Hence, $\cN=\{x \in\IR^\nu \mid \, 1_{\{ 0 \}}(B_x)\neq 0 \}$ is a nullset. Note that if $A_{\Lambda,x}$ converges to a selfadjoint operator $C_x$  in the strong resolvent sense as $\Lambda\to\infty$ for some $x\in \RR^\nu$, then $B_x=(C_x-z)^{-1}$ is injective, so $x\notin \cN$.
	
	Let $f:\RR\mapsto\RR$ be defined by $f(0)=\lambda-z$ and $f(x)=1/x$ for $x\neq 0$. Then
	\begin{equation*}
	A= f((A-z)^{-1})+z=  I_{\oplus,x}( f(B_x)+z  ),
	\end{equation*}
	by the injectivity of $(A-z)^{-1}$. For $x\in\IR^\nu$, we define the selfadjoint operator $A_x=f(B_x)+z$. Then the family $\{A_x\}_{x\in\IR^\nu}$ is strong resolvent measurable by \cref{Thm:Dirint}. As $0\leq B_x\leq \frac{1}{\lambda-z}$, we see $A_x\geq \lambda$. Further, $B_x=(A_x-z)^{-1}$ on $\cN^c$, so $\{A_x\}_{x\in\IR^\nu}$ is resolvent continuous on $\cN^c$ and $A_{\Lambda,x}$ converges to $A_x$ in the norm resolvent sense as $\Lambda\to\infty$ on $\cN^c$.
\end{proof}
\noindent
We end this appendix with a lemma about the domains of direct integrals, which we use in the proof of theorem \cref{LargeFiberProp}.
\begin{lem}\label{lem:dirIntDomain}
	Let $\{A_x\}_{x\in \RR^\nu}$, $\{B_x\}_{x\in \RR^\nu}$ be strong resolvent measurable families of selfajoint operators, $A=I_{\oplus,x}(A_x), B=I_{\oplus,x}(B_x)$ and assume $\{A_x\}_{x\in \RR^\nu}$ is uniformly bounded below. Further, let $z<\inf_{\RR^\nu}\inf(\sigma(A_x))$.\\
	If $\cD(A)\subset \cD(B)$, then $\cD(A_x)\subset \cD(B_x)$ for almost all $x\in \RR^\nu$ and $x\mapsto B_x(A_x-z)^{-1}$ is essentially bounded. If further $\cD(A_x)\subset \cD(B_x)$ for all $x$, then $x\mapsto B_x(A_x-z)^{-1}$ is strongly measurable and $B(A-z)^{-1}=I_{\oplus,x}(B_x(A_x-z)^{-1})$.
\end{lem}
\begin{proof}
	Note $z<\inf\sigma(A)$ by \cref{Thm:Dirint}.
	
	Now let $\cD \subset \fH$ be countable and dense and write $C=\lVert B(A-z)^{-1} \lVert$. We want to prove there is a nullset $F$ such that $\cD\subset \cD(B_x(A_x-z)^{-1})$ for $x\in F^c$ and $\lVert B_x(A_x-z)^{-1}\psi \lVert\leq C\lVert \psi \lVert$ for all $\psi\in \cD$. As $\cD$ is countable, this amounts to proving that for any $\psi\in \cD$ we have $\psi\in \cD(B_x(A_x-z)^{-1})$ almost everywhere and $\lVert B_x(A_x-z)^{-1}\psi \lVert\leq C\lVert \psi \lVert$ almost everywhere.
	To that end, we define the heat kernel
	\begin{equation*}
	\phi_{t,y}(x)=(4\pi t)^{-\nu/2}\exp(-\lvert x-y\lvert^2/4t ) \qquad\mbox{for}\ t>0\ \mbox{and}\ y\in\IR^\nu.
	\end{equation*}
	Then we have $(A+z)^{-1}\phi_{t,y}^{1/2} \phi_{1,0}^{1/2} \psi\in \cD(B)$, which implies $$\phi_{t,y}(x) \phi_{1,0}(x) (A_x-z)^{-1}\psi\in \cD(B_x) \qquad\mbox{for almost all}\ x\in\IR^\nu.$$
	As $\phi_{t,y} \phi_{1,0}>0$, we see $\psi\in \cD(B_x(A_x-z)^{-1})$ almost everywhere. Further, we have
	\begin{align*}
	\int_{\RR^\nu} \phi_{t,y}(x) \phi_{1,0}(x)\lVert B_x(A_x-z)\psi \lVert^2 dx&=\lVert B(A-z)^{-1} \phi_{t,y}^{1/2} \phi_{1,0}^{1/2} \psi\lVert^2\\&\leq C^2 \int_{\RR^\nu} \phi_{t,y}(x) \phi_{1,0}(x) dx\lVert \psi \lVert^2.
	\end{align*}
	Integrating a function $g\in L^1(\RR^\nu)$ against the heat kernel $\phi_{t,y}(x)$ gives a function in the $y$-variable that converges to $g$ in $L^1(\RR^\nu)$ as $t\to0$ (see e.g. \cite[Theorem 7.19]{Grigoryan.2009}). As $\phi_{1,0}>0$, we see $\lVert B_y(A_y-z)^{-1}\psi \lVert \leq C \lVert \psi\lVert$ for almost every $y\in \RR^\nu$.
	
	Assume now $\cD(A_x)\subset \cD(B_x)$ for all $x\in \RR^\nu$. Note that for any $\psi\in \fH$ we have $x\mapsto B_x(A_x-z)^{-1}\psi$ is measurable by \partref{Thm:Dirint}{measureAble}, so $I_{\oplus,x}(B_x(A_x-z)^{-1})$ is well defined. That $I_{\oplus,x}(B_x(A_x-z)^{-1})=B(A-z)^{-1}$ is obvious.
\end{proof}

\section{Proof of Pull-Through Formula}
\label{app:pullthrough}

This appendix is devoted to proving the pull through formula. The method is built on the results in \cite{DamMoller.2018a} and the reader should consult this paper for the proofs. We start by defining
\begin{equation*}
\cF_{+}=  \bigtimes_{n=0}^{\infty}\FS^{(n)}
\end{equation*}
with coordinate projections $P_n$. We equip $\cF_{+}$ with the $\sigma$-algebra induced by the projections. Define
\begin{equation*}
\cC(\RR^\nu)=\{ f:\RR^\nu\rightarrow \FS_+ \mid \forall n\in\IN_0: x\mapsto P_nf(x)\in L^2(\RR^\nu,\FS^{(n)})\}/\sim.
\end{equation*}
where $f\sim g$ if an only if $f=g$ almost everywhere. We will now introduce the pointwise annihilation operator. For $\psi=(\psi^{(n)}) \in \cF_+$ we define $A\psi\in \cC(\RR^\nu)$ by
\begin{equation*}
P_n(A\psi)(k)=a_k\psi^{(n+1)},
\end{equation*}
where the pointwise annihilation operator $a_k$ is defined as in \cref{eq:annihilator}. Since $a_k\psi^{(n+1)}\in\FS^{(n)}$ is well-defined for almost all $k\in\IR^\nu$, one easily observes this defines a continuous operator $A:\FS_+\to \cC(\IR^\nu)$. 
The next statement can be found in \cite{DamMoller.2018a}
\begin{thm}\label{Thm:CalculatingSecondQuantisedUsingAnihilation}
	Let $B:\IR^\nu\to\IR$ be measurable with $B\ge 0$.
	Then $$\psi\in \cD(d\Gamma(B )^{\frac{1}{2}})\iff  B^{\frac{1}{2}} A\psi\in L^2(\RR^\nu, \FS)$$ Furthermore, for $\phi,\psi\in \cD(d\Gamma(B )^{\frac{1}{2}})$ we have
	\begin{equation}\label{formula2}
	\langle  d\Gamma(B)^{\frac{1}{2}}\phi,d\Gamma(B)^{\frac{1}{2}}\psi\rangle=\int_{\cM}B(k) \langle A\phi(k),A\psi(k) \rangle  d\mu(k),
	\end{equation}
	and $A\psi(k)\in \cF$ almost everywhere on $\{k\in\IR^\nu:B(k)>0 \}$.
\end{thm}
\noindent
A pull-through formula for the cutoff Hamiltonian is proven in \cite[Lemma B.12]{Dam.2018}.
\begin{thm}\label{Lem:pullthrCutoff}
	Let $\omega$ and $v$ satisfy \cref{hyp1,,hyp2,,hyp3}. Further, let $\xi \in \RR^\nu$, $\nu\geq 2$ and $\Lambda \in (0,\infty)$. Assume $\psi\in\cD(H_\Lambda(\xi))$ satisfies $ A(H_\Lambda(\xi)-\Sigma_\Lambda(\xi))\psi(k)\in \FS$ for almost every $k\in\IR^\nu$. Then
	\begin{align*}
	(A\psi)(k)=&(H_\Lambda(\xi-k)-\Sigma_\Lambda(\xi)+\omega(k))^{-1}(A(H_\Lambda(\xi)-\Sigma_\Lambda(\xi))\psi)(k)\\&-v_\Lambda(k)(H_\Lambda(\xi-k)-\Sigma_\Lambda(\xi)+\omega(k))^{-1}\psi
	\end{align*}
	almost everywhere.
\end{thm}
\noindent
We can now prove the pull-through formula for the renormalized case.
\begin{thm}\label{Lem:pullthr}
	Let $\omega$ and $v$ satisfy \cref{hyp1,,hyp2,,hyp3}. Further, let $\xi \in \RR^\nu$, $\nu\geq 2$. Assume $\psi$ is a ground state for $H_\infty(\xi)$. Then
	\begin{align*}\label{PullthroughFormula}
	(A\psi)(k)=-v(k)(H_\infty(\xi-k)+\omega(k)-\Sigma_\infty(\xi) )^{-1}\psi
	\end{align*}
	almost everywhere.
\end{thm}
\begin{proof}
	Define $E_\infty=0$ and define for $\Lambda \in [0,\infty]$ the operator
	\begin{align*}
	H_\Lambda' (\xi,k)=H_\Lambda(\xi-k)-\Sigma_\Lambda(\xi)=H_\Lambda(\xi-k)+E_\Lambda-(\Sigma_\Lambda(\xi)+E_\Lambda)
	\end{align*}
	Similar to the proof of \cref{lem:Q-formula}, we find $H_\Lambda' (\xi,k)$ converges to $H_\infty' (\xi,k)$ in the norm resolvent sense. 
	Pick $\chi \in C_0^\infty(\RR)$ such that $\chi(0)=1$ and pick $\Lambda_0$ such that $\Sigma_\Lambda(\xi)+E_\Lambda-\Sigma_\infty(\xi)+1>0$ for all $\Lambda>\Lambda_0$. We then define for $a\in\{0,1\}$, $\Lambda\in(\Lambda_0,\infty)$
	\begin{align*}
	\psi_\Lambda&=\chi(H_\Lambda'(\xi,0))\psi,\\
	C_{a,\Lambda}&=H_\Lambda'(\xi,0)^a(H_\Lambda'(\xi,0)+E_\Lambda+\Sigma_\Lambda(\xi)-\Sigma_\infty(\xi)+1)\chi(H_\Lambda'(\xi,0)),\\
	B&=(d\Gamma(\omega)+1 )^{1/2},\\
	D_{\Lambda}&=B (H_\Lambda(\xi)+E_\Lambda-\Sigma_\infty(\xi)+1)^{-1}.
	\end{align*}
	By the spectral theorem, \cref{prop:operatorsareSA,LargeFiberProp}, we see $H_\Lambda'(\xi,0)^a\psi_\Lambda\in \cD(H_\Lambda(\xi))\subset \cD(B)$ and that
	\begin{align*}
	B(H'_\Lambda(\xi,0))^a\psi_\Lambda=D_\Lambda C_{a,\Lambda}\psi.
	\end{align*} 
	We abuse notation by setting $0^0=1$. By the spectral theorem, the norm resolvent convergence of $H_\Lambda'(\xi,0)$, \partref{lem:normresolventConv}{convCalc} and $\chi(H_\infty'(\xi,0))\psi=\psi$, we see $C_{a,\Lambda}\psi$ converges to $0^a\psi$ as $\Lambda\to\infty$. Thus, using \cref{LargeFiberProp}, we find $D_\Lambda C_{a,\Lambda}\psi$ converges to $0^a\psi$. Hence $(H'_\Lambda(\xi,0))^a\psi_\Lambda$ converges to $0^a\psi$ in $B$-norm.
	
	By \cref{Thm:CalculatingSecondQuantisedUsingAnihilation} we see that $AH_\Lambda' (\xi,0)\psi_\Lambda $ is Fock space valued, so we may apply \cref{Lem:pullthrCutoff} and find
	\begin{align*}
	(A\psi_\Lambda)(k)=&(H_\Lambda(\xi-k)-\Sigma_\Lambda(\xi)+\omega(k))^{-1}(AH'_\Lambda (\xi,0)\psi_\Lambda)(k)\\&-v_\Lambda(k)(H_\Lambda(\xi-k)-\Sigma_\Lambda(\xi)+\omega(k))^{-1}\psi_\Lambda.
	\end{align*}
	By \cref{Thm:CalculatingSecondQuantisedUsingAnihilation} we see that $\omega^{1/2}A\psi_\Lambda$ converges to $\omega^{1/2}A\psi$ in $L^2(\RR^\nu,\cF(\cH))$. Additionally, it holds that $\omega^{1/2}AH'_\Lambda (\xi,0)\psi_\Lambda$ converges to 0 in $L^2(\RR^\nu,\cF(\cH))$. As $\omega>0$ almost everywhere, we may pick elements $\Lambda_0<\Lambda_1<\Lambda_2<\cdots$ such that 
	\begin{align*}
	&\lim_{n\rightarrow \infty} (A\psi_{\Lambda_n})(k)=(A\psi)(k)\\
	&\lim_{n\rightarrow \infty} (AH'_{\Lambda_n} (\xi,0)\psi_{\Lambda_n})(k)=0
	\end{align*}
	for almost every $k$. Now
	\begin{align*}
	\lim_{\Lambda\rightarrow \infty}(H_\Lambda(\xi-k)-\Sigma_\Lambda(\xi)+\omega(k))^{-1}=(H_\infty(\xi-k)-\Sigma_\infty(\xi)+\omega(k))^{-1}
	\end{align*}
	in norm except at $k\in \RR \xi$ (see \partref{lem:resolventformbound}{lem:lowerboundsigma2}), which finishes the proof.
\end{proof}

\bibliographystyle{halpha-abbrv}
\bibliography{lit}

\end{document}